\documentclass[11pt]{article}

\usepackage{fullpage}
\usepackage{algorithm,algorithmicx,algpseudocode,relsize,bm}
\usepackage{amsfonts,amsthm,amsmath,amssymb,mathtools,float}
\usepackage{array}
\usepackage{epsfig}
\usepackage{url}
\usepackage{fullpage}
\usepackage{appendix}

\def\calc{{\cal C}}

\newtheorem{theorem}{Theorem}
\newtheorem{lemma}{Lemma}
\newtheorem{corollary}{Corollary}
\newtheorem{thm}{Theorem}[section]

\theoremstyle{definition}
\newtheorem{defn}[thm]{Definition}

\theoremstyle{remark}

\title{Parallel (and other) extensions of the deterministic online model for bipartite matching and max-sat}
\author{N.~Pena\thanks{Dept of Computer Science,University of Toronto, pena@cs.toronto.edu.} \and A.~Borodin\thanks{Dept of Computer Science,University of Toronto, bor@cs.toronto.edu.} }

\begin{document}

\maketitle

\begin{abstract}
The surprising results of Karp, Vazirani and Vazirani \cite{Karp1990a} and (respectively) 
Buchbinder et al \cite{Buchbinder2012} are examples where rather simple randomization provides provably better approximations than the corresponding 
deterministic counterparts for online bipartite matching and (respectively) 
unconstrained non-monotone submodular. We show that seemingly strong extensions of the deterministic online 
computation model can at best match the performance of naive randomization. 
More specifically, for bipartite matching, we show that in the priority model 
(allowing very general ways to order the input stream), we cannot improve upon the trivial $\frac{1}{2}$ approximation achieved by any greedy maximal matching algorithm and likewise cannot improve upon this approximation by any $\frac{log n}{\log \log n}$ number of online algorithms running in parallel. The latter result yields an improved $\log \log n - \log \log \log n$ lower bound for the number of advice bits needed. For max-sat, we adapt the recent de-randomization approach of Buchbinder and Feldman \cite{Buchbinder2016} applied to the Buchbinbder et al \cite{Buchbinder2012} algorithm for max-sat to obtain a deterministic $\frac{3}{4}$ approximation algorithm using width $2n$ parallelism.  In order to improve upon this approximation, we show that 
exponential width parallelism of online algorithms is necessary (in a  model 
that is more general than what is needed for the width $2n$ algorithm).   

\end{abstract}

\section{Introduction}

It is well known that in the domain of online algorithms it is often 
provably necessary to use randomization in order to achieve reasonable approximation 
(competitive) ratios. It is interesting to ask when can the use of randomization be replaced by extending the online framework. In a  more constructive sense, can we de-randomize certain online algorithms by 
considering more general one pass algorithms? This question has already been answered in a couple of senses. 
B\"{o}ckenhauer et al \cite{Bock11} show that a substantial class of randomized algorithms can be 
 transformed (albeit {\it non-uniformly} and inefficiently) to an online algorithm with (small) advice. Buchbinder and Feldman \cite{Buchbinder2016} show how to uniformly and efficiently de-randomize the Buchbinder et al 
\cite{Buchbinder2012} algorithm for the unconstrained non-monotone submodular function maximization (USM) problem. The resulting de-randomized algorithm can be viewed as a 
parallel algorithm in the form of a ``tree of online algorithms''. We formalize their algorithm as an online restriction of
the Alekhnovich et al \cite{Alekhnovich2011a} pBT model. 

In this paper we consider two classical optimization problems, namely 
maximum cardinality bipartite matching and the max-sat problem. 
As an offline problem, it is well known that graph matching and more specifically bipartite matching (in both the unweighted and weighted cases) can be solved optimally in polynomial time. Given its relation to the adwords problem for online advertising, bipartite matching has also been 
well-studied as a (one-sided) online problem. Max-sat is not usually thought of as an online problem but many of the combinatorial approximation algorithms  
for max-sat (e.g. Johnson's algorithm \cite{Johnson}, Poloczek and Schnitger \cite{Poloczek2011c}, and  
Buchbinder et al \cite{Buchbinder2012}) can be viewed as online or more generally one-pass myopic algorithms.   

To study these types of problems from an online perspective, several precise models of computation have been defined. With respect to these models, we can begin to understand the power and limitations of deterministic and randomized algorithms that in some sense can be viewed as online or one-pass. Our paper will be organized as follows. We will conclude this section with an informal list of our main results.  
The necessary definitions will be provided  in Section \ref{sec:preliminaries}. Section \ref{sec:related} contains a review of the most relevant 
previous results. Section \ref{sec:max-sat} considers online parallel width results for max-sat. In Section \ref{sec:bipartite} we will consider 
results for bipartite matching with respect to parallel width, the priority model \cite{Borodin2003}, and the random order model (ROM). We conclude with a number of open problems in 
Section \ref{sec:conclusions}. 

\subsection{Our results}
\begin{itemize}

\item
For the max-sat problem, we will show that the Buchbinder and Feldman \cite{Buchbinder2016} de-randomization method can be applied to obtain a deterministic parallel width $2n$ 
online $\frac{3}{4}$ approximation algorithm.

\item We will then show (in a model more general than what is needed for the 
above $\frac{3}{4}$ approximation), that exponential width is required 
to improve upon this ratio even for the exact max-2-sat problem for
which Johnson's algorithm already achieves a $\frac{3}{4}$ approximation. 

\item We also offer a plausible width 2 algorithm that might achieve the $\frac{3}{4}$ approximation  or at least might improve upon the $\frac{2}{3}$ approximation  
achieved by Johnson's deterministic algorithm \cite{Johnson} which in some models is provably the best online (i.e. width 1) algorithm.

\item For bipartite matching we show that constant width (or even width $\frac{\log n}{\log \log n}$) can cannot asymptotically beat the trivial $\frac{1}{2}$ approximation
 acheived by any greedy maximal matching algorithm. This implies that more than 
$\log \log n - \log \log \log n$ advice bits are needed to asymptotically improve upon the $\frac{1}{2}$ approximation achieved by any greedy maximal matching algorithm.

\item We offer a plausible candidate for an efficient polynomial width 
online algorithm for bipartite matching. 

\item
For bipartite matching,
we will also show that the ability to sort the input
items as in the priority model cannot compensate for absense of randomization.

\item
We also make some observations about bipartite matching in the random order model.

\end{itemize}

\section{Preliminaries: definitions, input and algorithmic Models}
\label{sec:preliminaries}

We will briefly describe the problems of interest in the paper and then proceed to define the algorithmic models and input models relative to which we will 
present our results. 

\subsection{The Bipartite Matching and Max-Sat Problems}
In the unweighted  matching problem, the input is a graph $G=(V,E)$ and the objective is to find the largest subset of edges $S \subseteq E$ that are vertex-disjoint 
and such that $|S|$ is as large as possible. Bipartite matching is the special case where $V = A \cup B$ and $E \subseteq A \times B$.

In the (weighted) max-sat problem, the input is a propositional formula in CNF form.
That is, there is a set of clauses $\calc = \{C_1, \ldots, C_n\}$ where each clause is a set of literals and each literal is a propositional variable or its negation. 
The objective is to find an assignment to the variables that maximizes the number of clauses in $\calc$ that are satisfied. In the weighted case, each clause has an associated weight and the objective is to maximize the sum of weights of clauses that are satisfied. 
The max-sat problem has been generalized to the submodular max-sat problem where there is a normalized monotone submodular function $f:2^{\calc} \rightarrow \mathbb{R}$ and we wish to assign the variables to maximize $f(\calc')$ where $\calc' \subseteq \calc$ is the subset of satisfied clauses.

When a specific problem needs to be solved, there are many possible input instances. An instance is just a specific input for the problem in hand. For example, in bipartite matching, an instance is a bipartite graph. In weighted max-sat, an instance is a set of clauses (including a name, weight, and the literals in each clause). An algorithm will have the goal of obtaining a good solution to the problem for every possible instance (i.e. we are only considering worst case complexity). To establish inapproximation results, we construct an adversarial instance or a family of bad instances for every algorithm (one instance or family of instances per algorithm). We will mainly study deterministic algorithms and their limitations under certain models.

We measure the performance of online algorithms by the competitive (or approximation) ratio\footnote{The name competitive ratio is usually used when considering online problems while approximation ratio is used in other settings. We will just use approximation ratio in any model of computation. For the maximization problems we consider, the approximation ratio is typically considered as a fraction less than or equal to 1.}. 

\begin{defn}
Let $\mathbb{A}$ be an algorithm for a problem. Let $\mathbb{{\cal I}}$ be the set of all possible instances for the problem. For $I \in \mathbb{{\cal I}}$, let $v(\mathbb{A}, I)$ be the value obtained by the algorithm $\mathbb{A}$ on instance $I$ and let $v(I)$ be the optimal value for that instance. The approximation ratio of algorithm $\mathbb{A}$ is:
\[ \inf_{I \in \mathbb{{\cal I}}} \frac{v(\mathbb{A}, I)}{v(I)}\]

Let ${\cal I}_n$ be the set of all instances of size $n$ for the problem. The asymptotic approximation ratio of $\mathbb{A}$ is:
\[ \liminf_{n \in \mathbb{N}} \inf_{I \in {\cal I}_n} \frac{v(\mathbb{A}, I)}{v(I)} \]
\end{defn}

The approximation ratio considers how the algorithm performs in every instance, while the asymptotic approximation ratio considers how the algorithm performs on large instances. Having a small instance where an algorithm performs poorly shows that the algorithm has a low approximation ratio, but says nothing about its asymptotic approximation ratio. 


\subsection{Algorithmic Models}

We define the precise {\it one-pass} algorithmic models that we consider in this paper. For each, the algorithm may receive some limited amount of information in advance. Other than this, an instance is composed of individual {\it data items}. We define the size of an instance as the number of data items it is made of. Data items will be received in a certain order, and how this ordering is chosen depends on the algorithmic model. In addition, when a data item is received, the algorithm must make an irrevocable decision regarding this data item before the next data item is considered. 
The solution then consists of the decisions that have been made. 
We do not make any assumptions regarding time or space constraints for our algorithms. In fact, the limitations proven for these algorithmic models are information-theoretic: since the algorithm does not know the whole instance, there are multiple potential instances, and any decision it makes may be bad for some of these. Of course, how each data model is defined will depend on the specific problem (and even for one problem there may be multiple choices for the information contained in data items). 
Unless otherwise stated, for inapproximations we assume that the algorithm knows the size, and for positive results we assume that it does not know the size.

\subsubsection{Online Model}

In the online model, the algorithm has no control whatsoever on the order in which the data items are received. That is, data items arrive in any order (in particular, they may arrive in the order decided by an adversary and as each data item arrives, the algorithm has to make its decision. Thus, in this model an adversary chooses an ordering to prevent the algorithm from achieving a good approximation ratio. As long as it remains consistent with previous data items, the data item that an adversary presents to the algorithm may depend on previous data items and the choices that the algorithm has made concerning them. The following presents the structure of an online algorithm. 

\medskip
\noindent\fbox{%
\begin{minipage}{\dimexpr\linewidth-2\fboxsep-2\fboxrule\relax}
\begin{algorithmic}[1]
\Statex \textbf{Online Algorithm}
\State On an instance $I$, including an ordering of the data items ($d_1, \ldots, d_n$):
\State $i:=1$ 
\State \textbf{While} there are unprocessed data items 
\State The algorithm receives $d_i$ and makes an irrevocable decision for $d_i$ 
\Statex (based on $d_i$ and all previously seen data items and decisions).
\State $i:= i+1$
\State \textbf{EndWhile}
\end{algorithmic}
\end{minipage}%
}
\medskip

In online bipartite matching, it is standard to consider that the algorithm knows in advance the names of the vertices from one of the sides of the graph, which we call the offline side. 
The other side, which we call the online side, is presented in parts. A data item consists of a single vertex from the online side along with all of its neighbours on the offline side. At each step, an algorithm can match the online vertex to any of its unmatched neighbours, or it can choose to reject the vertex (leave it unmatched). In either case, we say that it processes the vertex. All decisions made by the algorithm are irrevocable. We call this problem online one-sided bipartite matching.

A data item consists of a single variable along with some information about the clauses where this variable appears. We consider four models, in increasing order of the amount of knowledge received:

\textbf{Model 0:} The data item is a list with the names and weights of the clauses where the variable appears positively and a list of names and weights of the clauses where the variable appears negatively.

\textbf{Model 1:} Model 0, plus the length of each clause is also included.

\textbf{Model 2:} Model 1, plus for each clause we also include the names of the other variables that occur in this clause. The data item does not include the signs of these variables.

\textbf{Model 3:} Model 2, plus for each clause we also include the signs of other variables in the clause. That is, in this model, the data item contains all the information about the clauses in which the variable occurs.

For submodular max-sat, in addition to a description of the model for variables and their clauses, we need to state how the submodular function is presented to the algorithm. Since the submodular function has domain exponential in the size of the ground set (in this case the number of clauses), it is usually assumed that the algorithm does not receive the whole description of this function from the start. Instead, there is an oracle which may answer queries that the algorithm makes. A common oracle model is the value oracle, where the algorithm may query $f(S)$ for any subset $S$. In our restricted models of computation, this is further restricted so that the set $S$ can include only elements seen so far, including the current element for which the algorithm is making a decision. For submodular max-sat, this means restricting to sets of clauses each containing at least one seen variable. However, for some algorithms this is too restrictive. In the double-sided myopic model from Huang and Borodin \cite{Huang2014}, the value oracle may also query the complement function $\overline{f}$ defined by $\overline{f}(S) = f(C \setminus S)$ (where $C$ is the ground set).

\subsubsection{Width models}

We consider a framework which provides a natural way to allow more powerful algorithms, while maintaining in some sense an online setting. The idea is that now, instead of maintaining a single solution, the algorithm keeps multiple possible solutions for the instance. When the whole instance has been processed, the algorithm returns the best among the set of possible solutions that it currently has. The point is to limit the number of possible solutions the algorithm can have at any point in time. The possible solutions can be viewed as a tree with levels. A node in the tree corresponds to a possible solution and an edge means that one solution led to the other one after the algorithm saw the next data item. At the beginning, there is a single empty solution: the algorithm has made no choices yet. This will be the root of the tree and it will be at level 0. Every time a possible solution is split into multiple solutions, the tree branches out. Every time a data item is processed, the level increases by 1. The algorithm may decide to discard some solution. This corresponds to a node with no children and we say that the algorithm cuts the node. The restriction is that the tree can have at most $k$ nodes in each level. Since each level represents a point in time, this means the algorithm can never maintain more than $k$ possible solutions.

There are three different models to consider here. In the max-of-$k$ model, the algorithm branches out at the very beginning and does not branch or cut later on. In particular, this model is the same as having $k$ different algorithms and taking the maximum over all of them in the end. In the width-$k$ model, the algorithm can branch out at any time. However, the algorithm cannot cut any possible solution: every node in the tree that does not correspond to a complete solution (obtained after viewing all data items) must have at least one child. Finally, in the width-cut-$k$ model, the algorithm can branch out and disregard a possible solution at any time (so that, at later levels, this possible solution will not contribute to the width count).

Although we shall not do so, the width models can also be extended to the priority and ROM settings. For example, we can consider width in the priority model as follows: the algorithm can choose an ordering of the data items as in the fixed priority model (in particular, the arrival order has to be the same for all partial solutions), and once an input item arrives the solutions are updated, branch, or cut as desired. This then is the 
fixed order pBT model of Alekhnovich et al \cite{Alekhnovich2011a}.
It is also possible to allow each branch of the pBT to adaptively choose the ordering of the remaining items and this then is the adaptive pBT model. 

\subsection{Relation to  advice and semi-streaming models}
\label{max-of-k-becomes}

Most inapproximation results for the online and related models we consider  
 allow the algorithm to know $n$, the size of
the input.
One would like to allow these algorithms to also know other
easily computed information about the input (e.g. maximum or minimum  
degree of a node, etc.)
In keeping with the information
theoretic nature of inapproximation results, one way to state
this is to allow these algorithms to use any small (e.g. $O(\log n)$)
bits of advice and not require that the advice be efficiently computed. 
The {\it online with advice model} \cite{Bock11, EmekFKR11} that is related to our width models allows the algorithm to access a (say) binary advice string initially given to it by an oracle that knows the whole input and has unlimited computational power. Clearly, the advice string could encode the optimal strategy for the algorithm on the input, so the idea is to understand how the performance of the algorithm changes depending on how many bits of advice it uses. 
A non-uniform algorithm is a set of algorithms, one for each value of $n$. In particular, a non-uniform algorithm knows $n$, the size of the input. The following simple observation shows that the non-uniform advice and max-of-$k$ models are equivalent.

\begin{lemma}\label{WidthAdvice}
Suppose that an online algorithm knows $n$, the size of the input. Then there is an algorithm using $b(n)$ advice achieving an approximation ratio of $c$ if and only if there is a max-of-$2^{b(n)}$ algorithm with approximation ratio $c$.
\end{lemma}
\begin{proof}
Let $\mathbb{A}$ be an online advice algorithm using advice of size $t$ and achieving a $c$ approximation ratio. The following max-of-$2^t$ algorithm achieves this ratio: try all possible advice strings of length $t$, and take the best option among these. If $\mathbb{M}$ is a max-of-$2^t$ algorithm with approximation ratio $c$, then there is a $t$  bit advice algorithm that, for each input, encodes the best choice among the $2^t$ that will achieve this approximation ratio.
\end{proof}

The online advice model and therefore the small width online model also has a weak relation 
with the graph
semi-streaming graph, a  model suggested by Muthukrishhan \cite{Muthukrishnan05} and studied  futher in Feigenbaum et al \cite{FeigenbaumKMSZ}. In that model, {\it edges} arrive online for a graph optimization problem (e.g. 
matching). More directly related to our models, Goel et al \cite{GoelKK12} consider the model where vertices arrive online. In either case, letting $n$ be the number of vertice, the algorithm is constrained to use 
${\tilde O}(n) = n \log^{O(1)}n$ space which 
can be  substantially less space than the number of edges.  
The online advice and semi-streaming models are not directly comparable since on the one hand semi-streaming algorithms are not forced to make irrevcable decisions, while online algorithms are not space constrained. 
In order to relate the online model to the semi-streaming model, we need to restrict the online model to those algorithms in which the computation satisifies   
  the ${\tilde O}(n)$ semi-streaming space bound. In particular, the 
algorithm 
cannot store all the information contained in the data items that have been considered in previous iterations. 
%
%
%

\subsubsection{Priority model}

The difference between the priority and online models is that in the priority model the algorithm has some power over the order in which data items arrive. Every time a new data item is about to arrive, an ordering on the universe of potential data items is used to decide which one arrives. This ordering is  do not impose any restrictions on how this order is produced (but since it is produced by the algorithm, it cannot depend on data items from the current instance that it has not yet seen). In particular, the ordering need not be computable. In the 
{\it fixed priority model}, the algorithm only provides one ordering at the beginning. For a given instance, the data items are shown to the algorithm according to this ordering (and as before, when one arrives, the algorithm must make a decision regarding the data item).

\medskip
\noindent\fbox{%
\begin{minipage}{\dimexpr\linewidth-2\fboxsep-2\fboxrule\relax}
\begin{algorithmic}[1]
\Statex \textbf{Fixed Priority Algorithm}
\State The algorithm specifies an ordering $\pi: U \rightarrow \mathbb{R}$, where $U$ is the universe of all possible data items.
\State On instance $I$, the data items are ordered $d_1, \ldots, d_n$ so that $\pi(d_1) \leq \pi(d_2) \leq \ldots \leq \pi(d_n)$.
\State $i:=1$
\State \textbf{While} there are unprocessed data items
\State The algorithm receives $d_i$ and makes a decision for $d_i$
\Statex (based on $d_i$ and all previously seen data items and decisions).
\State $i:= i+1$
\State \textbf{EndWhile}
\end{algorithmic}
\end{minipage}%
}
\medskip

In the {\it adaptive priority model}, the algorithm provides a new ordering every time a data item is about to arrive. Thus, in the fixed order model, the priority of each item is a (say) real valued function of the input item and in the adaptive order model, the priority function can also depend on all previous items and decisions. In the matching problem, an algorithm could for example choose an ordering $\pi$ so that data items corresponding to low degree vertices are preferred. Or it could choose $\pi$ to prefer data items corresponding to vertices that are neighbours of some specific vertex. Unless otherwise stated, when we say priority we mean adaptive priority. The following shows the template for adaptive priority algorithms:

\medskip
\noindent\fbox{%
\begin{minipage}{\dimexpr\linewidth-2\fboxsep-2\fboxrule\relax}
\begin{algorithmic}[1]
\Statex \textbf{Adaptive Priority Algorithm}
\State On instance $I$, initialize $D$ as the set of data items corresponding to $I$ and $U$ as the universe of all data items. Let $n := |D|$.
\State $i:=1$
\State \textbf{While} there are unprocessed data items
\State The algorithm specifies an ordering $\pi: U \rightarrow \mathbb{R}$.
\State Let $d_{i} := \operatorname{argmin}_{d \in D}(\pi(d))$.
\State The algorithm receives $d_{i}$ and makes a decision about it.
\State Update: $D \leftarrow D \setminus \{d_{i}\}$, $U\leftarrow \{\mbox{data items consistent with }d_1, \ldots, d_i\}$.
\State $i:= i+1$
\State \textbf{EndWhile}
\end{algorithmic}
\end{minipage}%
}
\medskip

Usually, the adversarial argument in this model is as follows: the adversary begins by choosing a subset $S$ from the universe of all data items. This will be the set of potential data items for the problem instance. Now, for every $s_1 \in S$, the algorithm could choose this data item as the first one by using an appropriate ordering $\pi$. After the algorithm chooses $s_1$ and makes its decision for this data item, the adversary further shrinks $S$, thus obtaining a smaller subset of potential data items. The algorithm then proceeds by choosing a second data item $s_2 \in S$ and makes a decision concerning it. After this the adversary further shrinks $S$, and this goes on until $S$ becomes empty.

In some adversarial instances, some data items may be indistinguishable to the algorithm. If $d$ and $d'$ are indistinguishable to the algorithm, the algorithm may produce an ordering $\pi$ to try to receive $d$, but the adversary can force $d'$ to be received instead. For instance, in the matching problem, in the beginning data items corresponding to vertices with the same degree will be indistinguishable. If the algorithm tries to produce an ordering $\pi$ to get the data item of a specific vertex, the adversary can rename vertices so that the algorithm receives the data item of some other vertex with the same degree. However, after the first vertex has been seen, the algorithm now does have some limited information about the names of other nodes and can possibly exploit this knowledge 
as we will see in Section \ref{subsec:randomized-priority}. 

In the most common model for general graph matching, a data item consists of a vertex name along with its neighbours. Here, every vertex has an associated data item. This contrasts with the common data item for online bipartite matching, where only the online side vertices have associated data items. We call the former model (restricted to instances that are bipartite) two-sided bipartite matching, and the latter one-sided. We shall restrict attention to the one-sided problem. 

\subsubsection{The random order model}

In the random order model, usually abbreviated by ROM, neither the algorithm nor the adversary chooses the order in which the data items are presented. Instead, given an input set chosen by an adversary, 
a permutation of the data items is chosen uniformly at random and this permutation dictates the order in which data items are presented. Once the random input permutation has been instantiated, the model is the same as the one-sided online model.

\begin{defn}
Let $\mathbb{A}$ be an online algorithm for a problem whose set of all possible instances is ${\cal I}$. For an instance $I$ of size $n$ and a permutation $\sigma \in S_n$, let $I(\sigma)$ be $I$ with data items presented in the order dictated by $\sigma$. Let $v(\mathbb{A}, I(\sigma))$ be the value achieved by the algorithm on $I(\sigma)$, and let $v(I)$ be the optimal value for $I$. The approximation ratio of $\mathbb{A}$ in ROM is:
\[ \inf_{I \in {\cal I}} \frac{\mathbb{E}_{\sigma \in S_n}[v(\mathbb{A},I(\sigma))]}{v(I)}\]

Let ${\cal I}_n$ be the set of instances of size $n$. Then the asymptotic approximation ratio of $\mathbb{A}$ in ROM is:
\[ \liminf_{n \in \mathbb{N}} \inf_{I \in {\cal I}_n} \frac{\mathbb{E}_{\sigma \in S_n}[v(\mathbb{A},I(\sigma))]}{v(I)} \]
\end{defn}

\section{Related work}
\label{sec:related}

The analysis of online algorithms in terms of the competitive ratio was 
explicitly begun by Sleator and Tarjan \cite{SleatorT85} although there were some previous papers that implictly were doing competitive analysis (e.g. 
Yao \cite{Yao80}). The max-of-$k$ online width model was introduced in Halld\'orsson et al \cite{Iwama2000} and Iwama and Taketomi \cite{Iwama2002} where they considered the the maximum independent set and knapsack problems. 
Buchbinder and Feldman showed a deterministic algorithm with approximation ratio $1/2$ for unconstrained submodular maximization which fits the online width model (but not the max-of-$k$ model) \cite{Buchbinder2016}. Their initial approach involved solving an LP at each step. They showed how to simplify the algorithm so that it does not require an LP solver, and the width used is linear.

Hopcroft and Karp \cite{HopcroftK73} 
showed that unweighted bipartite matching can be optimally solved offline in time $O(m \sqrt{n})$. For sparse graphs, the first improvement in 40 years is the ${\tilde O}(m^{\frac{10}{7}})$ time algorithm due to Madry \cite{Madry13}.   
With regard to the online setting, the seminal paper of 
Karp, Vazirani and Vazirani \cite{Karp1990a} established a number of surprising results for (one-sided) online bipartite matching. After 
observing that no online deterministic algorithm can do better than a $\frac{1}{2}$ approximation, they studied randomized algorithms.  In particular they showed that the natural randomized algorithm RANDOM (that matches an online vertex uniformly at random to an available offline vertex) only achieved an asymptotic approximation ratio of 
$\frac{1}{2}$;  that is, the same approximation as any greedy maximal matching algorithm. 
They then showed that their randomized RANKING algorithm achieved \footnote{It was later discoved 
\cite{GoelM08} (and independently by Krohn and Varadarajan) that there was an error in the KVV analysis. A correct proof was prodvided in \cite{GoelM08} and subsequently alternative proofs \cite{Birnbaum2008,Devanur2013a} have been provided.}
ratio $1-\frac{1}{e} \approx .632$. RANKING initially chooses a permutation of the offline vertices and uses that permutation to determine how to match an online vertex upon arrival. By deterministically fixing any permutation of the offline vertices, the Ranking algorithm can be interpreted as a deterministic algorithm in the ROM model. While Ranking is optimal as an online randomized algorithm, it is not known if its interpretation as a deterministic ROM algorithm is optimal for all 
online deterministic algorithms. Goel and Mehta \cite{GoelM08} show that no deterministic ROM algorithm can achieve an approximation better than $\frac{3}{4}$.
 The KVV algorithm can also be implemented as a 
$O(n \log n)$ space randomized semi-streaming algorithm whereas Goel el al 
\cite{GoelKK12} show that there is a determinstic semi-streaming 
algorithm using 
only $O(n)$ space provably establishing the power of the semi-streaming model. 
In the ROM model, the randomized Ranking  algorithm achieves an approximation ratio of at least $0.696 > 1-1/e$ \cite{Karande2011, Mahdian2011a} and at most $0.727$ \cite{Karande2011}.  
Following the KVV paper there have been a number of extensions of online bipartite matching with more direct application to online advertising (see, for example, \cite{Mehta2005, GoelM08, Aggarwal, Kesselheim2013a}),  and has also been studied in various stochastic models where the input graph is generated by sampling i.i.d from a known or unknown distribution of online vertices (see \cite{Feldman2009a, Bahmani2010, Manshadi2011, Jaillet, Karande2011}). Manshadi et al \cite{Manshadi2011} showed that no randomized ROM algorithm can achieve an asymptotic approximation ratio better than $0.823$ by establishing that inapproximation for the stochastic unknown i.i.d model.    

In the online with advice model for bipartite matching, D{\"{u}}rr et al \cite{Durr2016a} apply  the B\"{o}ckenhauer et al de-randomization idea to show that 
for every $\epsilon >0$, there is an ({\it inefficient}) $O(\log n)$ advice algorithm achieving ratio $(1-\epsilon)(1-\frac{1}{e})$. This is complemented by 
Mikkleson's \cite{Mikkelsen15} recent result showing that no online 
(even randomized)  algorithm 
using sublinear  $o(n)$ advice  can asymptotically improve upon the 
$1- \frac{1}{e}$ ratio achieved by KVV.  Furthermore, D{\"{u}}rr et al show that $O(\frac{1}{\epsilon^5} n) $ advice is sufficient, and $\Omega(\log(\frac{1}{\epsilon}) n$ advice is necessary  to achieve a $(1-\epsilon)$ ratio.
They show that for a natural but restricted class of algorithms, $\Omega(\log \log \log n)$ advice bits are  needed for deterministic algorithms to obtain an approximation ratio asymptotically better than $1/2$.
Finally, their Category-Advice algorithm is a deterministic two pass online 
(i.e. adversarial order) $\frac{3}{5}$ approximation 
algorithm where the first pass is used to give priority in the second pass  to 
the  offline vertices that were unmatched in 
the first pass. That is, the first pass is efficiently constructing an $n$ bit 
advice string for the second pass.  

The priority setting was introduced by Borodin et al \cite{Borodin2003}. It has been studied for problems such as makespan scheduling \cite{Angelopoulos2010a}, and also in several graph optimization problems \cite{Borodin2010a, Davis2009}. In terms of the maximum matching problem, most results are about general graph matching. Here, a data item consists of a vertex (the vertex that needs to be matched) along with a list of its neighbours. Aronson et al \cite{Aronson95} showed that the algorithm which at each step chooses a random vertex and then a random neighbour (to pick an edge to add to the matching) achieves an approximation ratio of $1/2+c$ for some $c>0$. Besser and Poloczek \cite{Besser2015} showed that MinGreedy, the algorithm that at each step picks an edge with a vertex of minimum degree, will not get an approximation ratio better than $1/2$ (even in the bipartite case), but for $d$-regular graphs this approximation ratio improves to $\frac{d-1}{2d-3}$. They showed that no deterministic greedy (adaptive) priority algorithm can beat this ratio for graphs of maximum degree $d$ (which implies that these cannot get an approximation ratio greater than $1/2$), and they showed no deterministic priority algorithm can get an approximation ratio greater than $2/3$
($5/6$ for ``degree based'' randomized algorithms). It should be noted that these inapproximability results for priority algorithms do not hold for the special bipartite case.

H\'astad \cite{Hastad01} showed that it is NP-hard to achieve an approximation ratio of $c$ for the maximum satisfiability problem for any constant $c > 7/8$, and the best known efficient algorithm has an approximation ratio of 0.797 and a conjectured approximation ratio of 0.843 \cite{Avidor2005}. The greedy algorithm that at each step assigns a variable to satisfy the set of clauses with larger weight is an online algorithm achieving an approximation ratio of $1/2$. Azar et al \cite{Azar2011a} observed that this is optimal for deterministic algorithms with input model 0. They showed a randomized greedy algorithm that achieves an approximation ratio of $2/3$ for online submodular max-sat, and they showed that this is optimal for input model 0. In this algorithm, when a variable arrives, the variable is set to true with probability $\frac{w_T}{w_T + w_F}$ and set to false otherwise, where $w_T$ is the weight of clauses satisfied if assigned to true and $w_F$ is the weight of clauses satisfied if assigned to false. For submodular max-sat, the weight is replaced by the marginal gain.

Johnson's algorithm \cite{Johnson} is a deterministic greedy algorithm that bases its decisions on the ``measure of clauses'' satisfied instead of the weights of these clauses. Yanakakis \cite{Yann} showed that Johnson's algorithm is the de-randomization (by the method of conditional expectations) of the naive randomized algorithm and also showed that no deterministic algorithm can achieve a better approximation ratio even in input model 3. Chen et al \cite{Chen1999a} showed that Johnson's algorithm achieves this $2/3$ approximation ratio. The analysis was later simplified by Engebretsen \cite{Engebretsen2004a}. Johnson's algorithm can be implemented in input model 1. Costello et al \cite{Costello2011} showed that Johnson's algorithm achieves an approximation ratio of $2/3+c$ for some $c>0$ in ROM. Poloczek and Schnitger gave an online randomized algorithm in input model 1 achieving an approximation ratio of $3/4$ \cite{Poloczek2011c}. They showed that Johnson's algorithm in ROM gets an approximation ratio of at most $2-\sqrt{15} < 3/4$, and they showed that the online randomized version of Johnson's algorithm (which assigns probabilities according to measures, as in the randomized greedy algorithm) achieves an approximation ratio of at most $17/23 < 3/4$.

Van Zuylen gave a simpler online randomized algorithm \cite{VanZuylen} with approximation ratio $3/4$. Buchbinder et al \cite{Buchbinder2012} gave a randomized algorithm for unconstrained submodular maximization with approximation ratio $1/2$ and additionally a related randomized algorithm for submodular max-sat achieving an approximation ratio of $3/4$. Poloczek \cite{Poloczek2011b} showed that no deterministic adaptive priority for max-sat can achieve an approximation ratio greater than $\frac{\sqrt{33}+3}{12} < 3/4$ in input model 2. He also showed that, under this input model, no randomized online algorithm can get an approximation ratio better than $3/4$, so several algorithms achieving an approximation ratio of $3/4$ that fit the framework are optimal (up to lower order terms). Yung \cite{YungPMS} showed that no deterministic priority algorithm for max-sat in input model 3 can achieve an approximation ratio better than $5/6$.  

By extending the online framework, Poloczek et al \cite{Polo2pass} proposed a deterministic algorithm achieving a $3/4$ approximation ratio that makes two passes over the input: in one pass, the algorithm computes some probabilities for each variable, and in the second pass, it uses these probabilities and the method of conditional expectations to assign variables.

A model for priority width-and-cut was presented and studied by Alekhnovich et al \cite{Alekhnovich2011a}. In particular, they showed that deterministic fixed priority algorithms require exponential width to achieve an approximation ratio greater than $21/22$ for max-sat in input model 3.

\section{Max-sat width results}

\label{sec:max-sat}

We will first show that the Buchbinder and Feldman \cite{Buchbinder2016} 
de-randomization approach can be utillizeid to obtain a $\frac{3}{4}$ 
approximation by a parallel online algorithm of width $2n$. Then we will show that with respect to what we are calling input model 2, that we would need exponential width to improve upon this approximation. Then we will propose a width 2 algotithm  as a plausible candidate to exceed Johnson's online $\frac{2}{3}$ approximation.

\subsection{Derandomizing the Buchbinder et al submodular max-sat algorithm}

Buchbinder et al \cite{Buchbinder2012}  presented a randomized algorithm for submodular max-sat with an approximation ratio of $3/4$. They define a \emph{loose assignment} of a set of variables $V$ as a set $A \subseteq V \times \{0,1\}$. Any variable can be assigned one truth value (0 or 1), none, or both. A clause is satisfied by $A$ if $A$ contains at least one of the literals in the clause. For instance, $V \times \{0,1\}$ will satisfy any clause and $\emptyset$ will satisfy no clause. Let $F$ be the normalized monotone submodular function on sets of clauses (this is part of the input to the problem) and let $g: \mathcal P \left( V \times \{0,1\} \right) \rightarrow \mathbb{R}$ be the function defined by $g(A) = F(C)$ where $C$ is the set of clauses satisfied by the loose assignment $A$. It is easy to check that $g$ is also a monotone submodular function.

The algorithm keeps track of a pair of loose assignments $(X,Y)$, which change every time a new variable is processed. Let $(X_i,Y_i)$ be the values right after the $i$th variable $v_i$ is processed. Initially the algorithm begins by setting $X_0 = \emptyset$ and $Y_0 = V \times \{0,1\}$. When processing $v_i$, $X_i$ will be $X_{i-1}$ plus an assignment $b$ to $v_i$, while $Y_{i}$ will be $Y_{i-1}$ minus the $1-b$ assignment to $v_i$. We say in this case that $v_i$ is assigned to $b$. Thus $X_i$ and $Y_i$ have the same unique assignment for the first $i$ variables, $X_i$ only contains assignments for the first $i$ variables, and $Y_i$ contains all possible assignments for the variables after $v_i$. If there are $n$ variables, $X_n = Y_n$ is a proper assignment, and this is the output of the algorithm.

When processing $v_i$, the algorithm makes a random decision based on the marginal gains of assigning $v_i$ to 0 and of assigning $v_i$ to 1. The value $g(X_{i-1} \cup \{(v_i,0)\}) - g(X_{i-1})$ is how much is gained by assigning $v_i$ to $0$, while $g(Y_{i-1}) - g(Y_{i-1} \setminus \{(v_i,1)\})$ is how much is surely lost by assigning $v_i$ to $0$. Thus, the quantity $f_i := g(X_{i-1} \cup \{(v_i,0)\}) - g(X_{i-1}) + g(Y_{i-1} \setminus \{(v_i,1)\}) - g(Y_{i-1})$ is a value measuring how favourable it is to assign $v_i$ to 0. Similarly, $t_i :=g(X_{i-1} \cup \{(v_i,1)\}) - g(X_{i-1}) + g(Y_{i-1} \setminus \{(v_i,0)\}) - g(Y_{i-1})$ measures how favourable the assignment of $v_i$ to 1 is. In the algorithm presented in \cite{Buchbinder2012}, $v_i$ is assigned to 0 with probability $\frac{f_i}{f_i + t_i}$ and to 1 with probability $\frac{t_i}{f_i + t_i}$ (with some care to avoid negative probabilities).

We now de-randomize this algorithm at the cost of having linear width. The de-randomization idea follows along the same lines as that of Buchbinder and Feldman \cite{Buchbinder2016} for a deterministic algorithm for unconstrained submodular maximization with a $1/2$ approximation ratio. The authors present the novel idea of keeping a distribution of polynomial support over the states of the randomized algorithm. Normally, a randomized algorithm has a distribution of exponential support, so the idea is to carefully choose the states that are kept with nonzero probability. Elements of the domain (or in our case, variables) are processed one at a time, and at each iteration a linear program is used to determine the changes to the distribution. They then argue that they can get rid of the LP's to obtain an efficient algorithm, since solving them reduces to a fractional knapsack problem. The same LP format used for unconstrained submodular maximization works for submodular maxsat (the only change in the linear program in our algorithm below are the coefficients), so the idea in \cite{Buchbinder2016} to get rid of the LP solving also works for our algorithm.

\begin{theorem}\label{submodularMS}
There is a linear-width double-sided online algorithm for submodular max-sat achieving an approximation ratio of $3/4$. The algorithm uses input model 1 of max- sat.
\end{theorem}
\begin{proof}
First, we note that an oracle for $F$ suffices for constructing an oracle for $g$. The algorithm keeps track of a distribution over pairs $(X,Y)$ of loose assignments of variables. A double-sided algorithm is needed to obtain the values of the $g(Y)$'s. The idea is to process the variables online, at each step changing the distribution. The pairs $(X,Y)$ satisfy the same properties as in the Buchbinder et al algorithm. Thus, $X$ corresponds to the assignments made in the partial solution so far, while $Y$ corresponds to this plus the set of potential assignments that the partial solution could still make. When all variables are processed, the support will contain proper assignments of variables, and the algorithm takes the best one. The distribution is constructed by using an LP (without an objective function) to ensure some inequalities hold while not increasing the support by too much. We use the notation $(p,X,Y) \in D$ to say the distribution $D$ assigns $(X,Y)$ probability $p$. Also, if $(X,Y) \in supp(D_{i-1})$, we use the notation $Pr_{D_{i-1}}[X,Y]$ to denote the probability of the pair $(X,Y)$ under distribution $D_{i-1}$. The variables are labelled $V = \{v_1, \ldots, v_n\}$ in the online order. See Algorithm~\ref{Submodular max-sat algorithm}.

\begin{algorithm}[!h]\caption{Submodular Max Sat}\label{Submodular max-sat algorithm}
\begin{algorithmic}[1]
	\State Let $D_0 = \{(1, \emptyset, V \times \{0,1\} )\}$
	\For{$i=1$ to $n$}
	\State $\forall (X,Y) \in supp(D_{i-1})$ let
\begin{eqnarray*}
	f_i(X,Y) &=& g(X \cup \{(v_i,0)\}) - g(X) + g(Y \setminus \{(v_i,1)\}) - g(Y)\\
	t_i(X,Y) &=& g(X \cup \{(v_i,1)\}) - g(X) + g(Y \setminus \{(v_i,0)\}) - g(Y)\\
	\end{eqnarray*}
	\State Obtain an extreme point solution for:
	\begin{eqnarray}
	\label{ineq1}\mathrm{E}_{D_{i-1}}[z(X,Y) f_i(X,Y) + w(X,Y) t_i(X,Y)] &\geq & 2\mathrm{E}_{D_{i-1}}[z(X,Y) t_i(X,Y)]\\
	\label{ineq2}\mathrm{E}_{D_{i-1}}[z(X,Y) f_i(X,Y) + w(X,Y) t_i(X,Y)] &\geq & 2\mathrm{E}_{D_{i-1}}[w(X,Y) f_i(X,Y)]\\
	z(X,Y) + w(X,Y) &=& 1 \mbox{ }\forall (X,Y) \in supp(D_{i-1})\label{probabilitieseq1}\\
	z(X,Y), w(X,Y) & \geq & 0 \mbox{ }\forall (X,Y) \in supp(D_{i-1})\label{ineq3}
	\end{eqnarray}
	\State Construct a new distribution:
	\begin{eqnarray*}
	D_i =& \{(z(X,Y) Pr_{D_{i-1}}[X,Y], X \cup \{(v_i,0)\}, Y \setminus \{(v_i, 1)\}): (X,Y) \in supp(D_{i-1}) \} \\
	 & \cup \{(w(X,Y) Pr_{D_{i-1}}[X,Y], X \cup \{(v_i,1)\}, Y \setminus \{(v_i, 0)\}): (X,Y) \in supp(D_{i-1}) \}
	\end{eqnarray*}
	\State\label{DelLA} Delete from $D_i$ any pair of loose assignments with zero probability.
	\EndFor\\
	\Return argmax$_{(X,Y) \in supp_{D_n}}\{g(X)\}$
\end{algorithmic}
\end{algorithm}

As before, $f_i(X,Y)$ is used to determine how profitable it is to assign $v_i$ to 0 in this pair, and similarly $t_i(X,Y)$ measures how profitable it is to assign $v_i$ to 1 in this pair. In normal max-sat, $f_i(X,Y)$ will be the weights of clauses satisfied by assigning $v_i$ to 0 minus the weights of clauses that become unsatisfied by this assignment (a clause becomes unsatisfied when it hasn't been satisfied and all of its variables have been assigned), and $t_i(X,Y)$ is the analogue for the assignment to 1.

Each $(X,Y) \in supp(D_{i-1})$ will potentially be split into two in $D_i$: $(X \cup \{(v_i,0)\}, Y \setminus \{(v_i, 1)\})$, corresponding to assigning $v_i$ to 0 in pair $(X,Y)$, and $(X \cup \{(v_i,1)\}, Y \setminus \{(v_i, 0)\})$, corresponding to assigning $v_i$ to 1 in pair $(X,Y)$. $z(X,Y)$ is the probability of assigning $v_i$ to 0, given $(X,Y)$. Similarly, $w(X,Y)$ is the probability of assigning $v_i$ to 1, given $(X,Y)$. Since at each step $D_i$ could grow twice the size, the LP is used to determine values for $z(X,Y)$, $w(X,Y)$ for all $(X,Y) \in supp(D_{i-1})$ such that the resulting distribution still satisfies the properties used to achieve a good assignment in expectation while forcing many of the variables to be 0. In Step~\ref{DelLA}, the pairs with zero probability are trimmed from $D_i$ to keep the distribution size small.

First, note that the distributions are well defined by induction and by inequalities~\ref{probabilitieseq1} and \ref{ineq3} of the LP. Also, $D_n$ contains well-defined assignments (instead of loose assignments), so the algorithm returns a valid assignment. Let $|D_i|$ be the size of the support of $D_i$. Excluding inequalities~\ref{ineq3} stating the non-negativity of variables, there are $|D_{i-1}|+2$ inequalities in the LP for step $i$, so an extreme point solution contains at most that many nonzero variables: $|D_i| \leq |D_{i-1}|+2$. Therefore, $|D_n| \leq 2n+1$ and the algorithm does have linear width.

Now, let us see that $t_i(X,Y) + f_i(X,Y) \geq 0$ for any $(X,Y) \in supp(D_{i-1})$. It can be proved by induction that for all $1 \leq i \leq n$ and all $(X,Y) \in supp(D_{i-1})$, $X \subseteq Y \setminus \{(v_i,1), (v_i,0)\}$. Then by submodularity,
\[g(X \cup \{(v_i,0)\}) - g(X) \geq g(Y) - g(Y \setminus \{(v_i,0)\})\]
\[g(X \cup \{(v_i,1)\}) - g(X) \geq g(Y) - g(Y \setminus \{(v_i,1)\})\]
Adding both inequalities, moving all terms to the left side and rearranging, we obtain that $t_i(X,Y) + f_i(X,Y) \geq 0$.

We now prove that for every $i$ the LP formed is feasible, which is assumed by the algorithm in order to find an extreme point solution. We give an explicit feasible solution:
\[ z(X,Y) = \frac{\max\{0, f_i(X,Y)\}}{\max\{0, f_i(X,Y)\} + \max\{0, t_i(X,Y)\}} \hspace{12pt} w(X,Y) = 1 - z(X,Y)\]
In case $t_i(X,Y)=f_i(X,Y)=0$, we take $z(X,Y)=1$ and $w(X,Y)=0$. By definition, equalities~\ref{probabilitieseq1} and inequalities~\ref{ineq3} hold. When $t_i(X,Y)=f_i(X,Y)=0$, the corresponding variables will not contribute to either side of inequalities \ref{ineq1} and \ref{ineq2}. Assume either $t_i(X,Y) \neq 0$ or $f_i(X,Y) \neq 0$ and we want to show inequality~\ref{ineq1} (since the other one will be analogous). Let $D:=\max\{0, f_i(X,Y)\} + \max\{0, t_i(X,Y)\}$. Then we want to show:
\[ \mathrm{E}\left[ \frac{\max\{0, f_i(X,Y)\}}{D} f_i(X,Y) + \frac{\max\{0, t_i(X,Y)\}}{D} t_i(X,Y)\right] \geq 2\mathrm{E}\left[\frac{\max\{0, f_i(X,Y)\}}{D} t_i(X,Y)\right] \]
Because $D>0$ this is equivalent to
\[ \mathrm{E}\left[ \max\{0, f_i(X,Y)\} f_i(X,Y) + \max\{0, t_i(X,Y)\} t_i(X,Y)\right] \geq 2\mathrm{E}\left[\max\{0, f_i(X,Y)\}t_i(X,Y)\right]\]
where the expectation is over $D_{i-1}$.

For $(X,Y)\in supp(D_{i-1}) $ for which $f_i(X,Y)<0$, we have $t_i(X,Y)>0$ and the inequality becomes $t_i(X,Y)^2 \geq 0$, which clearly holds. Similarly, when $t_i(X,Y)<0$ we must have $f_i(X,Y)>0$ and the inequality becomes $f_i(X,Y)^2 \geq 2f_i(X,Y)t_i(X,Y)$ which is true because the right hand side is negative. Finally, when $f_i(X,Y)\geq 0$ and $t_i(X,Y) \geq 0$, the inequality becomes $f_i(X,Y)^2+t_i(X,Y)^2 \geq 2 f_i(X,Y)t_i(X,Y)$, which is true because $(a-b)^2 = a^2+b^2-2ab \geq 0$.

Let $OPT$ be an optimal assignment. For any $1\leq i \leq n$ and $(X,Y) \in supp(D_i)$, let $OPT_i(X,Y) := (OPT \cup X) \cap Y$: it is an assignment that coincides with $X$ and $Y$ in the first $i$ variables and coincides with $OPT$ in the rest. We will now prove the following:

\begin{lemma} For $1 \leq i \leq n$:
\[\mathrm{E}_{D_{i-1}}[g(OPT_{i-1}(X,Y))] - \mathrm{E}_{D_{i}}[g(OPT_i(X,Y))] \leq \frac{1}{2}\left( E_{D_i}[g(X) + g(Y)] - E_{D_{i-1}}[g(X) + g(Y)]\right)\]
\end{lemma}
\begin{proof}
First suppose that in OPT, $v_i$ is assigned 0. In this case:
\begin{eqnarray*}
\mathrm{E}_{D_{i-1}}[g(OPT_{i-1}(X,Y))] &-& \mathrm{E}_{D_{i}}[g(OPT_i(X,Y))]  =  \sum_{(X,Y) \in supp(D_{i-1})} [ Pr_{D_{i-1}}[X,Y] g(OPT_{i-1} (X,Y))\\
& - & z_i(X,Y) Pr_{D_{i-1}}[X,Y] g\left(OPT_i(X \cup \{(v_i, 0)\}, Y \setminus \{(v_i,1)\}) \right)\\
& - & w_i(X,Y) Pr_{D_{i-1}}[X,Y] g\left(OPT_i(X \cup \{(v_i, 1)\}, Y \setminus \{(v_i,0)\}) \right) ]\\
& = & \sum_{(X,Y) \in supp(D_{i-1})} w_i(X,Y) Pr_{D_{i-1}}[X,Y] [g(OPT_{i-1}(X,Y))\\
& - & g\left(OPT_i(X \cup \{(v_i, 1)\}, Y \setminus \{(v_i,0)\}) \right) ]
\end{eqnarray*}

Here, we use $z_i$ and $w_i$ to emphasize that these are the extreme point solutions obtained at the $i$th LP. The first equality holds by construction of $D_i$. The second holds because $z_i(X,Y) = 1 - w_i(X,Y)$ and because for all $(X,Y) \in supp(D_{i-1})$, $OPT_{i-1}(X,Y) = OPT_i(X \cup \{(v_i, 0)\}, Y \setminus \{(v_i,1)\})$ since $(v_i,0) \in OPT$.

If $(X,Y) \in supp(D_{i-1})$, then $X \subseteq OPT_{i-1}(X,Y) \setminus \{(v_i,0)\} \subseteq Y \setminus \{(v_i,1)\}$ and by submodularity:
\[g(OPT_{i-1}(X,Y)) - g(OPT_{i-1}(X,Y) \setminus \{(v_i,0)\})  \leq  g(X \cup \{(v_i,0)\}) - g(X)\]
\[g(Y) - g(Y \setminus \{(v_i,1)\})  \leq g(OPT_{i-1}(X,Y) \setminus \{(v_i,0)\} \cup \{(v_i,1)\}) - g(OPT_{i-1}(X,Y) \setminus \{(v_i,0)\})\]

Adding these two inequalities, rearranging, and using the fact that $OPT_{i-1}(X,Y) \setminus \{(v_i,0)\} \cup \{(v_i,1)\} = OPT_i(X \cup \{(v_i, 1)\}, Y \setminus \{(v_i,0)\}) $ we obtain:
\begin{eqnarray*}
g(OPT_{i-1}(X,Y)) & -& g(OPT_i(X \cup \{(v_i, 1)\}, Y \setminus \{(v_i,0)\})) \\
& \leq & g(X \cup \{(v_i,0)\}) - g(X) - g(Y) + g(Y \setminus \{(v_i,1)\})\\
& = & f_i(X,Y)
\end{eqnarray*}

Thus, we conclude:
\begin{eqnarray*}
\mathrm{E}_{D_{i-1}}[g(OPT_{i-1}(X,Y)] &-& \mathrm{E}_{D_{i}}[g(OPT_i(X,Y)]\\
& \leq &\sum_{(X,Y) \in supp(D_{i-1})} w_i(X,Y) Pr_{D_{i-1}}[X,Y] f_i(X,Y)\\
& = & \mathrm{E}_{D_{i-1}} [ w_i(X,Y) f_i(X,Y)]
\end{eqnarray*}

Analoguously, when $(v_i,1) \in OPT$ we obtain:
\[ \mathrm{E}_{D_{i-1}}[g(OPT_{i-1}(X,Y)] - \mathrm{E}_{D_{i}}[g(OPT_i(X,Y)] \leq \mathrm{E}_{D_{i-1}} [ z_i(X,Y) t_i(X,Y)]\]

On the other hand, we have:
\begin{eqnarray*}
\frac{1}{2} ( \mathrm{E}_{D_i} [g(X) &+& g(Y)] - \mathrm{E}_{D_{i-1}} [g(X) + g(Y)] ) = \frac{1}{2} ( \mathrm{E}_{D_{i-1}} [z_i(X,Y) g(X \cup \{(v_i,0)\})\\
& +&  w_i(X,Y) g(X \cup \{(v_i,1)\}) + z_i(X,Y) g(Y \setminus \{(v_i,1)\}) + w_i(X,Y) g(Y \setminus \{(v_i,0)\})\\
& -& g(X) - g(Y) ] )\\
& = & \frac{1}{2} ( \mathrm{E}_{D_{i-1}}[ z_i(X,Y)(g(X \cup \{(v_i,0)\}) + g(Y \setminus \{(v_i,1)\}) - g(X) - g(Y))\\
& + & w_i(X,Y)(g(X \cup \{(v_i,1)\}) + g(Y \setminus \{(v_i,0)\}) - g(X) - g(Y))])\\
& = & \frac{1}{2}( \mathrm{E}_{D_{i-1}}[z_i(X,Y)f_i(X,Y) + w_i(X,Y)t_i(X,Y)])\\
& \geq & \max\{ \mathrm{E}_{D_{i-1}} [ w_i(X,Y) f_i(X,Y)], \mathrm{E}_{D_{i-1}} [ z_i(X,Y) t_i(X,Y)]\}
\end{eqnarray*}
where the first equality is by how $D_i$ is constructed and the last inequality is because of inequalities \ref{ineq1} and \ref{ineq2} from the LP.
\end{proof}

To conclude the proof of the theorem, we add the inequalities given by the lemma for $1 \leq i \leq n$, obtaining:
\[ \mathrm{E}_{D_0}[g(OPT_0(X,Y))] - \mathrm{E}_{D_n}[g(OPT_n(X,Y))] \leq \frac{1}{2}\left( E_{D_n}[g(X) + g(Y)] - E_{D_0}[g(X) + g(Y)]\right) \]

Notice that $\mathrm{E}_{D_0}[g(OPT_0(X,Y))] = g(OPT)$, $E_{D_0}[g(X)] = g(\emptyset)$, $E_{D_0}[g(Y)] = g(V \times \{0,1\})$, and for all $(X,Y) \in supp(D_n)$, $X = Y = OPT_n(X,Y)$. Therefore the inequality becomes
\[ g(OPT) - \mathrm{E}_{D_n}[g(X)] \leq \frac{1}{2} \left( 2E_{D_n}[g(X)] - g(\emptyset) - g(V \times \{0,1\})\right)\]

Therefore, after rearranging we get:
\[ \mathrm{E}_{D_n}[g(X)] \geq \frac{1}{2} g(OPT) + \frac{1}{4}[ g(\emptyset) + g(V \times \{0,1\}) ] \geq \frac{3}{4} g(OPT)\]
The last inequality follows from the fact that $F$ is normalized (so $g(\emptyset) = 0$) and monotone (so $g(V \times \{0,1\}) \geq g(OPT)$). Calling the algorithm's ouput assignment $A$, we conclude that
\[F(A) = g\left(\underset{(X,Y) \in supp_{D_n}}{\operatorname{argmax}} \{g(X)\}\right) \geq  \mathrm{E}_{D_n}[g(X)] \geq \frac{3}{4} g(OPT) = \frac{3}{4} F(OPT)\]

We note that the LP format is the same as that in \cite{Buchbinder2016}. The only difference with their LP is the coefficients. So their argument that this can be solved by viewing it as a fractional knapsack problem still holds.
\end{proof}

\subsection{Online width inapproximation bounds for max-2-sat}

We now present width impossibility results for max-sat with respect to different  input models. The best known efficient algorithm for max-sat has an approximation ratio of 0.797 \cite{Avidor2005}. Recall that Johnson's algorithm for max-sat \cite{Johnson} achieves a $2/3$ approximation ratio \cite{Chen1999a} and only requires the algorithm to know the lengths of the clauses; i.e. input model 1. Even for input model 2, we show in Theorem~\ref{maxsat2i} that exponential width-cut is required to improve upon the $3/4$ approximation ratio achieved by Algorithm~\ref{Submodular max-sat algorithm}, which is a linear width algorithm in input model 1. In Theorem~\ref{maxsat0}, we show that constant width algorithms cannot achieve an approximation ratio of 2/3 in input model 0. This shows that constant width is unable to make up for the power lost if the algorithm does not know the lengths of clauses or if the algorithm is required to be deterministic;  
we note that the randomized algorithm using probabilities proportional to weights achieves an approximation ratio of $2/3$ \cite{Azar2011a}. Finally, in Theorem ~\ref{maxsat3} we show that, in input model 3, exponential width is required to achieve an approximation ratio greater than $5/6$.

Our impossibility results hold even in some special cases of max-sat. In max-$q$-sat, the instance is guaranteed to have clauses of length at most $q$. Exact max-$q$-sat is the case where all clauses are of length exactly $q$. In the following theorem, we show that for input model 2, exponential width-cut cannot achieve a better approximation ratio than that achieved in Theorem~\ref{submodularMS}, even for exact max-2-sat. It should be noted that a $3/4$ approximation ratio is already achieved by the naive randomized algorithm (that sets a variable to $0$ (or $1$) with probability $1/2$) and by its de-randomization, Johnson's algorithm, for exact max-2-sat.

We say that a max-of-$k$ algorithm \emph{assigns} (or \emph{sets}) $x$ to $(b_1, b_2, \ldots, b_k) \in \{0,1\}^k$ if it assigns $x$ to $b_1$ in its first assignment, it assigns $x$ to $b_2$ in its second assignment, etc.

\begin{theorem}\label{maxsat2i}
For any $\epsilon > 0$ there exists $\delta > 0$ such that, for $k<e^{\delta n}$, no online width-cut-$k$ algorithm can achieve an asymptotic approximation ratio of $3/4 + \epsilon$ for unweighted exact max-2-sat with input model 2.
\end{theorem}
\begin{proof}
First, we will show a concrete example where any width-2 algorithm achieves an approximation ratio of at most $3/4$. Then we show a way to extend this to a $3/4+\epsilon$ asymptotic inapproximation with respect to the max-of-$k$ model for $k$ that is exponential in the number of variables. Finally, we briefly argue why this impossibility result will also hold in the more general width-cut-$k$ case.

Suppose $k=2$. The adversary begins by showing variable $x_1$: it appears positively in one clause and negatively in another clause, both of length 2. The remaining variable in both clauses is $y$. If the algorithm does not branch or sets the variable $x_1$ to $0$ or to $1$ in both assignments, the adversary can force a $3/4$ approximation ratio as follows. Suppose without loss of generality that in both assignments $x_1$ is assigned $1$. Then the adversary presents the instance: $(x_1 \vee \overline{y}) \wedge (\overline{x_1} \vee y) \wedge (\overline{y} \vee z) \wedge (\overline{y} \vee \overline{z})$. No assignment where $x_1$ is set to 1 can satisfy all clauses, but an assignment where $x_1$ and $y$ are set to $0$ satisfies all clauses. Thus, the algorithm achieves an approximation ratio of at most $3/4$.

Therefore the algorithm must set one assignment to $1$ and the other to $0$. We assume that the algorithm sets $x_1$ to $(1,0)$. Now the adversary presents a variable $x_2$, where again there is one clause where it appears positively and one where it appears negatively, both of length 2 and where the remaining variable is $y$. Then there are four cases depending on the decision of the algorithm on $x_2$ (in each case, the whole instance consists of four clauses in total):

\begin{tabular}{ l | l l l l }
  Decision on $x_2$ & $(0,0)$ & $(1,0)$ & $(1,1)$ & $(0,1)$ \\
  \hline
  Clause with $x_1$ & $x_1 \vee \overline{y}$ & $x_1 \vee \overline{y}$ & $x_1 \vee \overline{y}$ & $x_1 \vee y$ \\
  Clause with $\overline{x_1}$ & $\overline{x_1} \vee \overline{y}$ & $\overline{x_1} \vee y$ & $\overline{x_1} \vee \overline{y}$ & $\overline{x_1} \vee \overline{y}$ \\
  Clause with $x_2$ & $x_2 \vee y$ & $x_2 \vee y$ & $x_2 \vee \overline{y}$ & $x_2 \vee y$ \\
  Clause with $\overline{x_2}$ & $\overline{x_2} \vee \overline{y}$ & $\overline{x_2} \vee \overline{y}$ & $\overline{x_2} \vee y$ & $\overline{x_2} \vee \overline{y}$ \\
\end{tabular}

In all cases, the algorithm will only be able to satisfy 3 out of the 4 clauses in any of its branches, but the instance is satisfiable, so the inapproximation holds.

Now let us show how to extend this idea to max-of-$k$. Let $\epsilon > 0$, take $\delta = 7 \epsilon^2$, so that $k < e^{\delta n} = e^{7 \epsilon^2 n}$. The adversary will present variables $x_1, \ldots, x_{n}$, where each $x_i$ appears in two clauses of length 2 and where the remaining variable is $y$: in one $x_i$ appears positively and in the other negatively. In fact, the two clauses will either represent an equivalence to $y$ (given by $x_i \vee \overline{y}$, $\overline{x_i} \vee y$) or an inequivalence to $y$ (given by $x_i \vee y$, $\overline{x_i} \vee \overline{y}$), but the algorithm does not know which is the case. If an assignment does not satisfy the (in)equivalence correctly, it will get only one of the two clauses (ie $1/2$ of the total).

Suppose that the algorithm maintains $k$ assignments, and suppose it makes assignments on $x_1, \ldots, x_{n}$. Then the algorithm can only maintain at most $k$ of the possible $2^{n}$ assignments. For a fixed assignment of $x_1, \ldots, x_{n}$, by Chernoff bounds, the probability that a uniformly random assignment agrees with the fixed one on at least $n/2+2\epsilon n$ variables is at most $e^{-8 \epsilon^2 n}$. Similarly, the probability that it agrees with the fixed assignment on at most $n/2 - 2\epsilon n$ variables is at most $e^{-8 \epsilon^2 n}$. Thus, by union bounds, the probability that any of the two possibilities occurs on any of the $k$ assignments maintained by the algorithm is at most $2ke^{-8 \epsilon^2 n}< e^{-\epsilon^2n +ln2}<1$. So there exists an assignment $A$ that agrees with every assignment maintained by the algorithm on more than $n/2-2\epsilon n$ but less than $n/2+2\epsilon n$ of the variables.

The adversary uses this assignment $A$ to determine the signs of $y$ in the clauses, which in turn determines for each $i$ whether $x_i$ is equivalent or inequivalent to $y$. If $A$ assigns $x_i$ to 1, then the adversary says $x_i$ is equivalent to $y$. If $A$ assigns $x_i$ to 0, then the adversary says that $x_i$ is inequivalent to $y$. Clearly, the set of clauses constructed is satisfiable. Fix one of the $k$ assignments maintained by the algorithm. If to complete this assignment the algorithm sets $y$ to 1, the number of (in)equivalences satisfied by the assignment is equal to the number of variables where $A$ and this assignment agree, which is less than $n/2+2\epsilon n$. On the other hand, if to complete the assignment the algorithm sets $y$ to 0, then the number of (in)equivalences satisfied is equal to the number of variables where $A$ and this assignment disagree, which again is less than $n/2+2\epsilon n$. Since an assignment that satisfies $q$ (in)equivalences will satisfy a $\frac{n+q}{2n}$ fraction of the clauses, the approximation ratio achieved by the algorithm is less than $3/4 + \epsilon$.

It is easy to see why this result will also hold for width-cut: the only decisions of the adversary that depend on the branching are made when the last variable $y$ is being processed (their signs are determined in each clause by assignment $A$). So the adversary can use the strategy that corresponds to the assignments of the algorithm right before $y$ is presented. Any branching or cutting made when deciding the assignments for $y$ are irrelevant: since it's the last variable, the algorithm should just assign $y$ to maximize the number of satisfied clauses in each assignment.
\end{proof}

\begin{theorem}\label{maxsat0}
For any constant $k$, the asymptotic approximation ratio achieved by any online width-$k$ algorithm for unweighted max-sat with input model 0 is strictly less than $2/3$.
\end{theorem}
\begin{proof}
We start by giving a max-of-$k$ inapproximation result, which is then easily extended to width. It should be noted that for this result we need to allow the adversary's final instance to contain repeated equal clauses.

First consider the case $k=2$. The adversary presents a variable $x_1$. There are two clauses where it appears positively and two where it appears negatively. Without loss of generality, there are two options: both assignments set $x_1$ to 1 or the first sets $x_1$ to 1 and the second sets $x_1$ to 0. In the former case, the adversary proceeds to say that the clauses containing $\overline{x_1}$ were of length one but the clauses containing $x_1$ had an additional variable $y$, which means both assignments satisfy half of the clauses. In the latter case, the adversary now presents variable $y_1$. It appears positively in one of the clauses where $\overline{x_1}$ appears and it appears negatively in the other clause where $\overline{x_1}$ appears. The value the second assignment gives to this variable is irrelevant since it already satisfied these clauses. Without loss of generality, assume the first assignment sets it to 1. Then the adversary presents a variable $z_1$, which occurs only positively in the clause where $y_1$ appears positively. Thus the values both assignments give to this variable are irrelevant. The first assignment satisfied three of the four clauses while the second only satisfied two of them. There is an optimal solution satisfying all four: set $x_1$ to 1, $y_1$ to 0, $z_1$ to 1. Now, the adversary repeats this process, but reversing the roles of the two assignments so that now the first assignment only gets two out of four clauses and the second gets three. Adding up, both assignments get five out of the eight clauses and the optimal value is 8, so we get a $5/8$ inapproximation.

For the general (but constant) $k$ case, we proceed by induction to prove that there is an adversary giving an inapproximation ratio strictly less than $2/3$. Recall that online maxsat ($k=1$) in this model cannot get an approximation ratio better than $1/2$ \cite{Azar2011a} (it is easy to extend this to an asymptotic inapproximation). Suppose there is an adversarial strategy for $i<k$. We now present a strategy for max-of-$k$. The adversary begins by presenting many variables $v_1, v_2, \ldots, v_n$, each of which will have many clauses where it appears positively and the same number of clauses where it appears negatively. For each $1 \leq i \neq j \leq n$, the clauses where $v_i$ and $v_j$ appear are disjoint. These will be all of the clauses of the instance: the adversary will not present any new clauses later on. It will only present additional variables contained within these clauses.

After decisions are made there will be $k$ assignments, each assigning a value $0$ or $1$ to each of the variables. For every $S \in \{0,1\}^k$, the adversary will recursively apply its strategies for max-of-$i$ (for values $i<k$) to $V_S$, the set of variables from $v_1, v_2, \ldots, v_n$ assigned to $S$, and to the clauses where variables in $V_S$ appear. More precisely, let $\mathbb{A}_j$ be the adversary for max-of-$j$. Then $\mathbb{A}_k$ will simulate $\mathbb{A}_i$ by ignoring some of the assignments (since now it can only consider $i$ of them). Variables created by $\mathbb{A}_i$ will be new variables. When $\mathbb{A}_i$ creates new clauses, $\mathbb{A}_k$ uses some of the clauses where variables in $V_S$ appear instead of creating new ones. If $S=(1,\ldots,1)$, for these variables the adversary says that the clauses where they appear positively have length two and include a new variable $y$ but the clauses where they appear negatively have length one. Thus the algorithm can only satisfy $1/2$ of these clauses but the set of clauses is satisfiable. The response is analogous if $S=(0,\ldots,0)$.

Now suppose that $S = (s_1, \ldots, s_k)$ contains $t$ 1's and $k-t$ 0's, for some for $0<t<k$. Let $S_1 := \{1 \leq j \leq k | s_j = 1\}$ and let $S_0:= \{1 \leq j \leq k | s_j = 0\}$, so $|S_1| = t$ and $|S_2| = k-t$. The adversary will roughly split $V_S$ into two parts, one $V_S^1$ of size $w|V_S|$ and the other $V_S^2$ of size $(1-w)|V_S|$, for a $w$ to be determined. In the first part, the adversary will say that the positive clauses (where variables in $V_S^1$ appear positively) were of length one, and it will simulate $\mathbb{A}_t$ using the negative clauses, so the decisions made on the $k-t$ assignments indexed by $S_0$ don't matter and the adversary only considers the decisions of the algorithm on the $t$ assignments indexed by $S_1$. In the second part, the adversary will say that negative clauses were of length one, and will simulate $\mathbb{A}_{k-t}$ using the positive clauses and considering only the $k-t$ assignments indexed by $S_0$. See Table~\ref{VS12Example} for an example. Let $r_1$ be the inapproximation ratio for max-of-$t$ and let $r_2$ be the inapproximation ratio for max-of-$(k-t)$. Let $a = (1+r_1)/2$ and $b=(1+r_2)/2$. Then the proportion of clauses satisfied by assignments indexed by $S_1$ is at most $aw + 1/2(1-w)$ and the proportion of clauses satisfied by assignments indexed by $S_0$ is at most $1/2w + b(1-w)$. We select $w$ to minimize the maximum between these two amounts, by equating these values:

\[aw + \frac{1}{2}(1-w) = \frac{1}{2}w + b(1-w) \]

We solve this equation, obtaining

\[w = \frac{b-\frac{1}{2}}{a+b-1}\]

Plugging back in into the equality and considering there is an optimal assignment satisfying all clauses, we obtain an inapproximation ratio of

\[\frac{2ab - \frac{1}{2}}{2(a+b-1)}\]

\begin{table}[!h]
\begin{tabular}{ l l l | l l l }
$V_S^1$ & & & $V_S^2$ & &\\
\hline
& $v_1$ & Clause of length 1 & & $v_4$ & Used to simulate $\mathbb{A}_{k-t}$\\
& $v_2$ & Clause of length 1 & & $v_5$ & Used to simulate $\mathbb{A}_{k-t}$\\
& $v_3$ & Clause of length 1 & & $v_6$ & Used to simulate $\mathbb{A}_{k-t}$\\
& $\overline{v_1}$ & Used to simulate $\mathbb{A}_t$ & & $\overline{v_4}$ & Clause of length 1\\
& $\overline{v_2}$ & Used to simulate $\mathbb{A}_t$ & & $\overline{v_5}$ & Clause of length 1\\
& $\overline{v_3}$ & Used to simulate $\mathbb{A}_t$ & & $\overline{v_6}$ & Clause of length 1
\end{tabular}
\caption{Example where $V_S$ consists of 6 variables, 3 in $V_S^1$ and 3 in $V_S^2$.}
\label{VS12Example}
\end{table}

Now, recall $r_1$, $r_2$ are inapproximation ratios for max-of-$k'$ for $k'<k$, so by the induction hypothesis $1/2 \leq r_1, r_2 < 2/3$, which implies $3/4 \leq a,b < 5/6$. Given these parameters, it can be shown that the above ratio gets a value strictly less than $2/3$. Since all ratios are less than $2/3$ regardless of $S$, the inapproximation obtained overall is less than $2/3$. Notice that we assume that we can neglect the gain obtained when $|V_S|$ is not large enough to apply the recursive strategy (hence the number of initial variables $n$ has to be large) and in addition we are assuming we can at least approximate $w$ accurately (hence the number of clauses per $v_i$ has to be large).

To extend this result to width $k$, we can begin by assuming that $k'=1$ where $k'$ is the maximum number of assignments the algorithm keeps. We start by applying the adversary for max-of-$k'$. If the adversary finishes before the algorithm does any splitting then we are done. Otherwise the algorithm splits to now maintain $k'' > k'$ assignments and we apply the adversary for max-of-$k''$, but using many more clauses so that we ensure that the clauses used by previous adversaries will be negligible when calculating the approximation ratio.
\end{proof}

\begin{theorem}\label{maxsat3}
For any $\epsilon > 0$ there exists $\delta > 0$ such that, for $k<e^{\delta n}$, no online width-cut-$k$ algorithm can achieve an asymptotic approximation ratio of $5/6 + \epsilon$ for unweighted max-2-sat with input model 3.
\end{theorem}
\begin{proof}
We use an argument similar to the one in Theorem~\ref{maxsat2i}, but with a different clause construction. Given $\epsilon >0$, let $\delta = 144 \epsilon^2$ so that $k < e^{\delta n} = e^{144 \epsilon^2 n}$. The instance will contain variables $x_1, \ldots, x_{n/2}$ and $y_1, \ldots, y_{n/2}$. The clauses will be $x_i \vee \overline{y_i}$, $\overline{x_i} \vee y_i$, and either $y_i$ or $\overline{y_i}$. The first two clauses are satisfied if and only if $x_i = y_i$. The last clause determines whether $y_i$ should be assigned to 0 or 1. Any such set of clauses will be satisfiable. The adversary presents the $x_i$'s, and the width-cut algorithm produces at most $k$ assignments of these variables.

For a fixed assignment of the $\frac{n}{2}$ $x_i$'s, the probability that a uniformly random assignment agrees with the fixed one on more than $n/4+6\epsilon n$ of the variables is at most $e^{-144 \epsilon^2 n}$. Therefore, there exists an assignment $A$ that agrees with each of the assignments maintained by the algorithm on at most $ n/4+6\epsilon n$ variables. The adversary now presents the $y_i$'s. It chooses to include clause $y_i$ in the instance if $x_i$ is set to 0 in assignment $A$, and it includes $\overline{y_i}$ if $x_i$ is set to 1.. Whenever an assignment does not agree with $A$ on $x_i$, it will satisfy at most two of the three clauses where $y_i$ appears in. Therefore, no assignment can satisfy more than a $\frac{1(1/4 + 6\epsilon) + 2/3(1/2 - (1/4 + 6\epsilon)))}{1/2} = 5/6 + \epsilon$ fraction of the clauses.
\end{proof}

\subsection{Candidate for a width 2 approximation algorithm}

Now we present a max-of-2 algorithm for max-sat. We were unable to prove impossibility results saying max-of-$k$ algorithms for constant $k$ cannot achieve an approximation ratio of $3/4$ in input model 1 or 2. Thus we suggest trying to use Johnson's algorithm in some way. In a max-of-2 algorithm, when processing a variable, we want to have a preference for the case in which the two assignments set a variable differently, since otherwise the $2/3$ inapproximation bound for online algorithms could be applied to the max-of-2 algorithm. In Algorithm~\ref{maxsat2} we present a formal way to do this, where the parameter $p$ controls how much we value different assignments. The variables are $\{x_1, \ldots, x_n\}$ in online order, and the algorithm constructs two assignments $A_1$ and $A_2$. Recall that, at any point in the algorithm, the measure of a clause $C$, $\mu(C)$, is defined by the product of its weight times $2^{-l}$ where $l$ is the number of variables not yet assigned.

We call the algorithm Width-2-Johnson's Algorithm. When deciding the assignment for $x_i$, it calculates, for each of the two assignments, the measures of clauses that become satisfied when assigning $x_i$ to 0 or to 1, as in Algorithm~\ref{JohnsonAlg}. For each assignment $D \in \{0,1\}^2$, it adds the measures of satisfied clauses when assigning $x_i$ to $D$. However, instead of double-counting clauses that become satisfied in both assignments, 
the measures of repeated clauses are multiplied by $p$ instead of added twice 
for $1 \leq p \leq 2$.

\begin{algorithm}\caption{Width-2-Johnson's with parameter $p$}\label{maxsat2}
\begin{algorithmic}[1]
\State Initialize $A_1$ and $A_2$ to be empty assignments.
\For{$i=1$ to $n$}
\State The algorithm will calculate a function $f: \{0,1\}^2 \rightarrow \mathbb{R}$
\For{$D = (d_1, d_2) \in \{0,1\}^2$} define
\State $C_1:$ the set of clauses that become satisfied by assigning $x_i$ to $d_1$ in $A_1$.
\State $C_2:$ the set of clauses that become satisfied by assigning $x_i$ to $d_2$ in $A_2$.
\State The algorithm sets $f(D) = \mu(C_1 \setminus C_2) + \mu(C_2 \setminus C_1) + p\mu(C_1 \cap C_2)$.
\EndFor
\State The algorithm assigns $x_i$ to argmax$_{D \in \{0,1\}^2}\{ f(D) \}$. 
\EndFor\\
\Return the assignment that satisfies a larger weight of clauses.
\end{algorithmic}
\end{algorithm}

Width-2-Johnson's algorithm is a candidate to achieve a good approximation ratio for max-sat. The following lemma suggests that $p=1.5$ is the right choice. It's unclear whether it could achieve an approximation ratio of $3/4$. It would be interesting to show that it achieves an approximation ratio greater than $2/3$, which would show that, unlike bipartite matching, max-sat is helped by constant width.

\begin{lemma}
Width-2-Johnson's algorithm cannot achieve an approximation ratio of $3/4$ if $p \neq 1.5$.
\end{lemma}
\begin{proof}
If $p>1.5$, let $w$ be such that $1> w>1/2$ and $p > 1+w$. Consider the max-sat instance consisting of clauses $x_1$ with weight 1, $\overline{x_1} \vee \overline{x_2}$ with weight $w$, and $x_2$ with weight $w$. The algorithm will assign $x_1$ to $(1,1)$ because $p > 1+w$. But then when processing $x_2$ it cannot satisfy both clauses of weight $w$ in any of the two assignments. Thus the approximation ratio achieved is $\frac{1+w}{1+2w}$, which is less than $3/4$ because $w > 1/2$.

If $p<1.5$, let $w$ be such that $0 < w < 1/2$ and $p < 1+w$. Consider the instance: $x_1$ with weight 1, $\overline{x_1}$ with weight $w$, $x_2$ with weight 1, $\overline{x_2}$ with weight $w$. The algorithm will assign $x_1$ to either $(1,0)$ or $(0,1)$ because $1+w > p$, so suppose it assigns $x_1$ to $(1,0)$. Then similarly suppose it assigns $x_2$ to $(0,1)$. Then both assignments satisfy clauses with total weight $1+w$, but the optimum assignment satisfies clauses of weight 2, hence the approximation ratio  achieved is $\frac{1+w}{2}$, and this is less than $3/4$ because $w<1/2$.
\end{proof}

\section{Bipartite matching results}
\label{sec:bipartite}


We will first consider width inapproximation results showing that width $\frac{\log n}{\log \log n}$ online algorithms cannot asymptotically improve upon the $\frac{1}{2}$ approximation given by any maximal matching algorithm.  We will then consider bipartite matching in the priority and  ROM  models. Our priority inapproximation shows that the ROM randomization cannot be
replaced by a judicious but deterministic ordering of the online vertices. 

\subsection{Width inapproximation}

We first fix  some notation for max-of-$k$ online bipartite matching. The algorithm keeps $k$ distinct matchings $M_1, \ldots, M_k$. Whenever an online vertex $u$ arrives, it can update each of the $M_j$'s by matching to $u$ one of its neighbours that $M_j$ has not yet matched. The size of the matching obtained by the algorithm is the maximum size of the $M_j$'s. We assume that the online vertices are numbered from $1$ to $n$, and the algorithm receives them in that order. The adversary chooses the offline vertices that are the neighbours. We will refer to the time when the algorithm chooses the matchings for the $i$-th online vertex as step $i$.

In the usual online bipartite matching problem, we can assume that the algorithm is greedy. This argument clearly still applies when we keep track of multiple matchings at once: we can assume that the algorithm is greedy in each. We begin with a max-of-$n^k$ algorithm and inapproximability result that follows from the relationship between advice and max-of-$k$ algorithms:

\begin{theorem}\label{WidthUppLow}
For every $\epsilon >0$ there exists a max-of-$n^{O(1)}$ algorithm achieving an approximation ratio of $1-1/e - \epsilon$. Also, no max-of-$2^{o(n)}$ algorithm can achieve an approximation ratio better than $1-1/e+\epsilon$.
\end{theorem}
\begin{proof}
By B\"{o}ckenhauer et al \cite{Bock11}, as observed in \cite{Durr2016a}, for every $\epsilon>0$ there is a $\Theta(\log n)$ advice algorithm achieving an approximation ratio of $1-1/e-\epsilon$. The algorithm is a de-randomization of the Ranking algorithm. Part of the advice string consists of an encoding of $n$. Even without this, the advice is still $\Theta(\log n)$. By Lemma \ref{WidthAdvice} there is a max-of-$2^{O(\log n)} = n^{O(1)}$ algorithm getting the desired approximation ratio. It should be noted that the algorithm uses an information-theoretic approach and is, in fact, extremely inefficient, in addition to requiring heavy pre-processing.

For the other part of the theorem, Mikkelsen \cite{Mikkelsen15} showed that, for every $\epsilon$, an advice algorithm with approximation ratio $1-1/e+\epsilon$ requires $\Omega(n)$ advice. If there was a max-of-$2^{o(n)}$ with this ratio, then by Lemma \ref{WidthAdvice} there would be a $o(n)$ advice algorithm achieving that ratio. Note that it can be assumed without loss of generality that any algorithm with $\Omega(\log n)$ advice knows $n$.
\end{proof}

We now prove some impossibility results concerning algorithms trying to beat the $1/2$ barrier that deterministic online algorithms cannot surmount. The adversarial graphs will be bipartite graphs with perfect matchings. The adversary will not only provide the graph but also construct a perfect matching ``online''. Once an offline vertex has been used in the adversary's perfect match, the adversary will not present it as a neighbour of any of the remaining online vertices. When an online vertex $u$ arrives, the adversary will choose a nonempty subset of offline vertices as its set of neighbours. Then the algorithm (which we assume without loss of generality to be greedy) chooses a match in each of the $k$ matchings. For each matching $M_i$, if there are neighbours of $u$ that have not been used in $M_i$, the algorithm must pick a neighbour $v$ and match $u$ to $v$ in $M_i$. When the algorithm has finished making its choices, the adversary picks one of the neighbours of $u$ and adds the pair of vertices to the perfect matching that it is constructing. The match to online vertex $i$ in this perfect matching is labelled as offline vertex $i$. We say that this offline vertex becomes \emph{unavailable}. An offline vertex is \emph{available} if it is not unavailable. The goal of the adversary is to force the algorithm to make as few matches as possible in the best of its $k$ matchings.

At a specific point in time and for any offline vertex $v$, we say $t(v)$ is the number of the algorithm's matchings that have used $v$. Whenever we say that the adversary \emph{gets rid} of an offline vertex $v$ at a given step we mean that, at this step, the only neighbour of the online vertex $u$ is $v$, so the best option of the algorithm is to match $u$ to $v$ in any of the matchings where $v$ has not yet been used. Also, $v$ will not be a neighbour of any of the remaining online vertices (the adversary must add $(u,v)$ to its perfect matching). If the adversary only gets rid of a constant number of offline vertices, the matchings made by the algorithm during these steps are negligible: they do not affect the asymptotic approximation ratio.

\begin{lemma}\label{Maxof2bm}
Any max-of-2 online bipartite matching algorithm cannot achieve a matching of size greater than $n/2 + 3$ on every input.
\end{lemma} 
\begin{proof}
Since the algorithm has only two matchings, at any point in time and for any offline vertex $v$, $0 \leq t(v) \leq 2$. There will be two stages. The first stage consists of steps where, at the beginning of the step, there are more than two available vertices $v$ with $t(v)=0$. The adversary chooses as neighbours all available vertices. Thus, we can guarantee that, after the algorithm has chosen matches for the online vertex of this step, there will still be at least one available vertex $v$ with $t(v)=0$. The adversary will choose one such offline vertex (i.e. one not chosen in any matching) to be added to its perfect matching. If at the beginning of a step, there are less than 3 available offline vertices with $t(v)=0$ (so that we cannot guarantee that there will exist a vertex with $t(v)=0$ after the algorithm does its matches), the adversary concludes stage 1 and gets rid of the at most 2 available offline vertices with $t(v)=0$ before stage 2.

Let $p$ be the number of steps that occurred during stage 1. Let $p_1$ and $p_2$ be the number of offline vertices with $t(v) = 1$ and $t(v)=2$ at the end of stage 1, respectively. During stage 1, at each step, $\sum_v t(v)$ is incremented by 2, so $2p = p_1 + 2p_2$. At each step in stage 1 a vertex becomes unavailable, but no vertex with $t(v) \neq 0$ becomes unavailable, so $n \geq p + p_1 + p_2$ (inequality because the adversary may get rid of vertices). Therefore, $p = \frac{p_1}{2} + p_2$ and $\frac{n}{2} \geq  \frac{3}{4} p_1 + p_2$. Right before the beginning of stage 2, the size of each matching is at most $p+2$: $p$ from stage 1 and 2 from getting rid of vertices.

The second stage consists of steps where at the beginning of the step all the available offline vertices satisfy $t(v) \geq 1$, and there are some available vertices with $t(v)=1$. At this stage, the $p_1$ vertices with $t(v)=1$ are considered. Since the number of vertices matched in $M_1$ and $M_2$ have to be the same during stage one, half of the vertices are used in $M_1$ and half in $M_2$. The steps in this stage will be either $M_1$-steps or $M_2$-steps (the order in which the adversary does them is irrelevant). An $M_1$ step is a step where the neighbours of the online vertex are available vertices that have already been matched in $M_1$, so the algorithm can only do a match in $M_2$. After the algorithm does its match, if there is still an available vertex that has only been matched in $M_1$, the adversary adds it to the perfect matching it is constructing, making it unavailable. At the beginning of stage 2, there are $p_1/2$ vertices with $t(v)=1$ matched in $M_1$. Every $M_1$-step (except possibly the last, where there may be only one available vertex matched in $M_1$ with $t(v)=1$) makes two vertices unusable: one because it is made unavailable and one because $t(v)$ changes from 1 to 2. So there will be $\lceil p_1/4 \rceil$ $M_1$-steps. We define $M_2$-steps analogously, and there will be $\lceil p_1/4 \rceil$ $M_2$-steps. After stage 2, all available vertices have $t(v)=2$ and are unusable, so no more matchings are made by the algorithm. The adversary can finish the construction of the perfect matching by making the remaining online vertices be neighbours of all offline vertices with $t(v)=2$.

The size of the matchings produced by the algorithm is at most $p + 2+ \lceil \frac{p_1}{4} \rceil$, since the algorithm can only increase the size of $M_1$ during stage 1, while the adversary gets rid of vertices, and during $M_2$-steps, and similarly for $M_2$. But $p + 2 + \lceil \frac{p_1}{4} \rceil \leq p+3+ \frac{p_1}{4} =  \frac{p_1}{2} + p_2 + 3 +\frac{p_1}{4} \leq \frac{n}{2} + 3$. This in particular implies that the asymptotic approximation ratio achieved by any algorithm is at most $\frac{1}{2}$.
\end{proof}

Clearly, making the width bigger without making $n$ bigger will eventually allow the algorithm to obtain an optimal matching using brute-force. However, it is natural to wonder whether by allowing $n$ to be large the adversary will be able to trick the algorithm into producing a small matching. This question is answered by the following:

\begin{theorem}\label{WidthBip1/2}
For any constant $k$, any width-$k$ online bipartite matching algorithm cannot achieve an asymptotic approximation ratio greater than $\frac{1}{2}$.
\end{theorem}
\begin{proof}
We will prove the theorem by first considering max-of-$k$ and then extending the result to width. We prove the following statement by induction on $k$: for any max-of-$k$ algorithm there exists a constant $c_k$ (that only depends on $k$) such that for every $n$ there is a graph $G$ of size $n$ ($n$ vertices on each of the two sides) where the algorithm obtains matchings of size at most $n/2 + c_k$. For $k=1$, the problem is the well studied online bipartite matching problem: there are adversarial graphs where we can take $c_1=1$ (needed for odd values of $n$). Lemma~\ref{Maxof2bm} proves the case for $k=2$, taking $c_2=3$. Suppose that the claim is true for max-of-$i$ for all $1 \leq i \leq k$, and we shall prove it for max-of-$(k+1)$. As before, the adversary will decide the neighbours of the incoming online vertex as well as the offline vertex that matches it in the perfect matching it constructs (and this offline vertex will not be a neighbour of any of the remaining online vertices). Let $M_1, \ldots, M_{k+1}$ be the matchings that the algorithm constructs.

The adversary will have the same first stage as in the lemma, at each step adding to its perfect matching a vertex with $t(v)=0$. When there are less than $k+2$ available vertices with $t(v)=0$ (so we cannot guarantee that there will be an available vertex with $t(v)=0$ after the algorithm does its matchings), stage 1 ends and then the adversary gets rid of all the available vertices with $t(v)=0$. Let $p$ be the number of steps in the first stage and let $q$ be the number of available offline vertices right after stage 1, ie the number of offline vertices not in the adversary's perfect matching at that time. Then $p + q = n$, and after the adversary gets rid of vertices there are at most $q$ available vertices. Now the adversary proceeds to a second stage.

For every $S \subseteq \{1, \ldots, k+1\}$ with $S \neq \emptyset$ and $S \neq \{1, \ldots, k+1\}$, let $V_S$ be the subset of available offline vertices that have been used in $M_i$ for all $i \in S$ and that have not been used in $M_i$ for all $i \notin S$. Let $S^C = \{1, \ldots, k+1\} \setminus S$, and notice that $0 < |S^C| < k+1$. At this point, only $M_i$'s with $i \in S^C$ can match vertices in $V_S$. The idea is that we will recursively apply our adversary for max-of-$|S^C|$ algorithms on a graph with $V_S$ as the set of offline vertices and with $|V_S|$ online vertices. The $M_i$'s with $i \in S$ are ignored: the algorithm cannot add matches in these when the set of neighbours of the online vertex is a subset of $V_S$. By the induction hypothesis, there is an adversarial strategy for max-of-$|S^C|$ such that the size of any of the matchings obtained (on a graph that uses $|V_S|$ online vertices and $V_S$ as the offline vertices) is at most $|V_S|/2 + c_{|S^C|}$. This will be close to 1/2 of the total when $|V_S|$ is large, since $c_{|S^C|}$ is a constant.

In the second stage, the adversary executes the max-of-$|S^C|$ strategies described above. For $S_1 \neq S_2$, the strategies will be independent because the set of offline neighbours is disjoint. Thus, the order in which the strategies are executed is irrelevant: they could even be executed in parallel. For concreteness, suppose the adversary first executes the strategies for subsets $S$ of size $k$ in lexicographic order (here it applies max-of-1 strategies),  then for subsets $S$ of size $k-1$ in lexicographic order (here it applies max-of-2 strategies), etc. After the strategies for all subsets have been executed, stage 2 is concluded and now we need to show that the adversary's perfect matching is about twice the size as any of the matchings constructed by the algorithm.

For simplicity, ignore the adversary getting rid of vertices and suppose that every $V_S$ is large enough. In the end, any fixed matching $M_i$ will use $p$ offline vertices because of stage 1. After applying our recursive adversaries, $M_i$ will use roughly half of the offline vertices that were not used by $M_i$ by the end of stage 1 but were still available at this time. The number of offline vertices that are available by the end of stage 1 is $q$. Thus, in the end, the size of matching $M_i$ is $p + \frac{q - p}{2} = \frac{p + q}{2} = n/2$.

Now we make the intuition from the previous paragraph precise. Notice that the size of matching $M_i$ is at most $p+k+1 +\sum_{S: i \notin S} \frac{V_S}{2}+c_{|S^C|}$: $p$ during stage 1, $k+1$ from vertices the adversary gets rid of before stage 2, and the rest during stage 2. The number of available offline vertices that $M_i$ has not used at the beginning of stage 2 is $\sum_{S: i \notin S} V_S \leq q - p = n - 2p$. By the induction hypothesis, $\sum_{S : i \notin S} c_{|S^C|} \leq \sum_{1 \leq j \leq k} \binom{k}{j-1} c_{j}$, since there are $\binom{k}{j-1}$ ways of choosing $S^C$ of size $j$ if we require $i \in S^C$. Therefore, the size of any matching obtained by the algorithm is at most $\frac{n}{2} + c_{k+1}$ where $c_{k+1} = k+1+\sum_{1 \leq j \leq k} \binom{k}{j-1} c_{j}$. This concludes the induction and the proof for max-of-$k$ algorithms.

Now we extend the result to width $k$. The idea is that we will slightly modify the adversary so that, given the decisions of the algorithm, for $a>b$, a max-of-$a$ adversary can be viewed as a max-of-$b$ adversary. The width-$k$ adversary will use this fact to change from the max-of-$b$ adversary to the max-of-$a$ adversary, whenever the algorithm branches, without affecting the argument. Let $\mathbb{A}_j$ be the max-of-$j$ adversary, but where the condition to end stage 1 is that there are less than $k+1$ available vertices with $t(v)=0$, instead of $j+1$. Also, we assume that $\mathbb{A}_j$ may perform the independent stage 2 simulations in any order we choose. The width-$k$ adversary $\mathbb{A}$ does the following: begin by assuming $m$, the maximum number of matchings maintained by the algorithm, is 1. When $\mathbb{A}$ needs to tell the algorithm which are the neighbours of the next online vertex, $\mathbb{A}$ does whatever $\mathbb{A}_m$ would do given the matchings the algorithm has made so far. If the algorithm does not branch, $\mathbb{A}$ constructs the perfect match as $\mathbb{A}_m$ would, and this finishes the processing of the online vertex. On the other hand, the algorithm may branch on the decisions of the online vertex, so that now it maintains $m+r$ matchings. Each new matching $M_{new}$ will branch off of some matching $M_{old}$, which in the branching tree means that now $M_{new}$ is a leaf of the subtree rooted at $M_{old}$ (or any of its ancestors). In this case, $\mathbb{A}$ simply increases $m$ by $r$. Then it simulates $\mathbb{A}_m$ ($m$ is the increased value) to obtain the perfect match. And this finishes the processing of the online vertex.

At any point in time, there is a max-of-$m$ algorithm that simulates the width-$k$ algorithm up to this point, if it knows the branching tree created up to this point. For each level of the branching tree (each corresponding to an online vertex), the max-of-$m$ algorithm keeps $l$ copies of each node, where $l$ is the number of leaves in the subtree rooted at that node. We claim that all previous decisions made by $\mathbb{A}$ are consistent with $\mathbb{A}_m$, in the following sense: the behaviour of $\mathbb{A}$ on the width-$k$ algorithm (which so far only has $m$ branches) is equivalent to the behaviour of $\mathbb{A}_m$ on the max-of-$m$ algorithm just described.  By behaviour of an adversary, we mean the offine neighbours it presents and the perfect matchings it constructs at each step.

We can show this by induction. Consider a step where the width-$k$ algorithm branches, and let $m_i$ and $m_f$ be the values of $m$ at the beginning and at the end of the step, respectively. Suppose that $\mathbb{A}$ (on the width-$k$ algorithm) behaves as $\mathbb{A}_{m_i}$ on the max-of-$m_i$ algorithm that simulates the width-$k$ algorithm up until the previous step. We will now show that $\mathbb{A}$ behaves as $\mathbb{A}_{m_f}$ on the max-of-$m_f$ algorithm that simulates the width-$k$ algorithm up until this step. On later steps, as long as the algorithm does not branch, this consistency will still hold. We will see $\mathbb{A}$ as $\mathbb{A}_{m_i}$ (on the max-of-$m_i$ algorithm) up until the end of the previous step, which is valid by our assumption. If the branching occurs during stage 1, then what we claim is true since $\mathbb{A}_{m_i}$ and $\mathbb{A}_{m_f}$ have the same stage 1.

Now suppose that the branching occurs during stage 2. We will prove that up until the beginning of the current step we can make the behaviour of $\mathbb{A}_{m_f}$ on the max-of-$m_f$ algorithm be the same as the behaviour of $\mathbb{A}_{m_i}$ on the max-of-$m_i$ algorithm. On a stage 2 step, $\mathbb{A}_{m_i}$ will be simulating $\mathbb{A}_x$ on $V_S$ for some $S \subseteq \{1, \ldots, m_i\}$ and where $x = |S^C|$. Then in $\mathbb{A}_{m_f}$ we choose to simulate a step of $\mathbb{A}_{y}$ on $V_{S'}$, where $S \subseteq S' \subseteq \{1,\ldots,m_f\}$, $S'$ contains all indices of matchings that branched off from matchings indexed in $S$, and $y = |{S'}^C| \geq |S^C|$ (${S'}^C$ contains indices of matchings that branched off from matchings indexed in $S^C$). This is because by definition the max-of-$m_f$ algorithm only keeps copies of the matchings that will later branch off, so $V_S$ = $V_{S'}$: here the left hand side corresponds to the set according to $\mathbb{A}_{m_i}$ and the right hand side is according to $\mathbb{A}_{m_f}$. More generally, there is a mapping $f: \mathcal{P} (\{1, \ldots, m_i\}) \rightarrow \mathcal{P}(\{1, \ldots, m_f\})$ that maps a set of indices $Q$ to the set of indices of matchings that branch off from matchings indexed in $Q$. Because of the behaviour of the algorithms, it holds that $V_Q$ according to $\mathbb{A}_{m_i}$ is equal to $V_{f(Q)}$ according to $\mathbb{A}_{m_f}$, and for any $R \subseteq\{1,\ldots,m_f\}$ that does not have a preimage under $f$, $V_R = \emptyset$. In the step where $m$ changes from $m_i$ to $m_f$, $\mathbb{A}_{m_i}$ simulating $\mathbb{A}_x$ to select a subset of $V_S$ as the set of neighbours is equivalent to $\mathbb{A}_{m_f}$ simulating $\mathbb{A}_y$ to select a subset of $V_{S'}$. After the algorithm does its decision and branching, $m$ is updated and $\mathbb{A}$ actually simulates $\mathbb{A}_{m_f}$, so the behaviour is the same. This concludes the proof of our claim.

Thus, in the end, the behaviour of $\mathbb{A}$ on the width-$k$ algorithm is equivalent to the behaviour of $\mathbb{A}_k$ on the max-of-$k$ algorithm. This means that the size of the matching constructed by the width-$k$ algorithm is at most $n/2 + c_k$. Since we changed stage 1 of the adversaries, the $c_k$'s will be slightly larger, but they still only depend on $k$.
\end{proof}

\begin{corollary}
Let $t(n) = \frac{\log n}{\log \log n}$. Any max-of-$t(n)$ online bipartite matching algorithm cannot achieve an asymptotic approximation greater than $\frac{1}{2}$.
\end{corollary}
\begin{proof}
For any $k$, from the proof of Theorem \ref{WidthBip1/2} max-of-$k$ algorithms can achieve matchings of size at most $n/2 + c_k$ on some hard graphs. First, we note that $c_k \leq k^k$. This is true for $k=1$. Assuming this holds for $1 \leq i < k+1$, then $c_{k+1} \leq k+1 + \sum_{1 \leq j \leq k} \binom{k}{j-1} j^j < \sum_{0 \leq j \leq k+1} \binom{k+1}{j} k^j = (k+1)^{k+1}$. In the second inequality we use the fact that $k+1 \leq 1+k^{k+1}$.

Now, notice that $t^t < (\log n)^{\frac{\log n}{\log \log n}} = (2^{\log \log n})^{\frac{\log n}{\log \log n}} = n$, where in the first inequality we omit dividing by $(\log \log n)^{\frac{\log n}{\log \log n}}$. This means that $t^t = o(n)$. Thus, max-of-$t$ algorithms achieve matchings of size at most $n/2 + t^t = n/2 + o(n)$.
\end{proof}

The following corrolary follow immediately from the observations in 
Section \ref{max-of-k-becomes}.
D{\"{u}}rr et al \cite{Durr2016a} proved an $\Omega(\log \log \log n)$ advice lower bound for achieving an approximation ratio greater than $1/2$, and this only applied to a restricted class of online advice algorithms. We improve this result:

\begin{corollary}
$\Omega(\log \log n)$ advice is required for an online algorithm to achieve an asymptotic approximation ratio greater than $1/2$ for bipartite matching, even when the algorithm is given $n$ in advance.
\end{corollary}
\begin{proof}
No $\log \left(\frac{\log n}{\log \log n}\right) = 
\log \log n - \log \log \log n$ advice algorithm can achieve an asymptotic approximation ratio better than $1/2$, even knowing $n$. Otherwise Lemma~\ref{WidthAdvice} would give a max-of-$\frac{\log n}{\log \log n}$ online algorithm achieving this ratio, contradicting the previous corollary.
\end{proof}

\subsection{Candidates for an approximation algorithm with polynomial width}

While the bounds for non-uniform max-of-polynomial bipartite matching are tight, we have not provided a simple and efficient constant width algorithm that achieves an approximation ratio better than $1/2$. We present two algorithms which could have an approximation ratio better than 1/2. The first candidate is 
a simple algorithm for max-of-$k$ bipartite matching that attempts to balance the current usage  of offline vertices that are still available. The second attempt tries to de-randomize the Ranking algorithm  based on the LP approach of 
Buchbinder and Feldman.

We consider the following simple max-of-$k$ algorithm for bipartite matching, for $k \in n^{O(1)}$. The idea works for max-of-$k$ for any number $k$, but the interesting question is what approximation ratio can we achieve when $k \in n^{O(1)}$. Let $M_1, \ldots, M_k$ be the matchings of the algorithm. Define the load of an offline vertex as the number of matchings in which it has been used. The algorithm does the following: when processing an online vertex, try to balance out the loads of the available offline vertices as much as possible. The intuition behind this is that we want the offline vertex with minimum load to have as high a load as possible, so if the adversary chooses to never show this vertex again, the number of matchings that will not match the vertex is as low as possible. The way of balancing loads can be made more precise or it could be left as an arbitrary choice for the algorithm. Nonetheless, this algorithm is fairly efficient and it would be interesting to see whether an algorithm that maintains polynomially many candidate matchings constructed this way can achieve an approximation ratio greater than $1/2$.

We have also considered Algorithm \ref{BipMLP} 
following the   
LP  approach of Buchbinder and Feldman.  
Suppose that the online vertices arrive in order $\{v_1, \ldots, v_n\}$. Let $N(v)$ be the set of neighbours of vertex $v$. For a matching $M$ and an online vertex $v$, let $N(M,v)$ be the set of offline vertices that are neighbours of $v$ and have not been used in $M$. Also, let $I(M) = \max\{i \in \{1, \ldots, n\} | v_i \mbox{ is matched in }M\}$, ie $I(M)$ is the maximum index of an online vertex matched in $M$. A permutation of the offline vertices is \emph{consistent} with $M$ if $M$ is the matching that results from performing $I(M)$ steps of Ranking given this permutation. In Algorithm~\ref{BipMLP}, we use an LP-based approach at each step to obtain a polynomial width algorithm for bipartite matching. We will consider distributions of matchings, and say $(p,M) \in D$ whenever the probability of $M$ under $D$ is $p$. For all $M \in D_{i-1}$ and $u \in N(M,v_i)$ let $P_i(u,M)$ be the probability that Ranking, when run with a random permutation consistent with $M$, will choose $(v_i,u)$ in the $i$th step. Let $S_{i-1}$ be the set of matchings $M$ in $D_{i-1}$ such that $|N(M,v_i)|>0$. The variables of the LP for the $i$th step are $x(u,M)$ for $u \in N(M,v_i)$ (when $(v_i,u)$ cannot be added to $M$, it doesn't make sense to have this variable). The intended meaning of $x(u,M)$ is that it is the probability in $D_i$ of matching $v_i$ to $u$, given $M$. Letting the probability of $M$ in $D_{i-1}$ be $Pr_{D_{i-1}}[M]$, then the probability of $M \cup \{(v_i, u)\}$ in $D_i$ will be $Pr_{D_{i-1}}[M] x(u,M)$.

\begin{algorithm}\caption{Bipartite Matching LP}\label{BipMLP}
\begin{algorithmic}[1]
        \State Let $D_0 = \{(1,\emptyset)\}$.
        \For{$i=1$ to $n$}
        \State Obtain an extreme point solution of the following set of inequalities:
        \begin{eqnarray}
                \mathbb{E}_{D_{i-1}}[x(u,M)] & \leq &  \mathbb{E}_{D_{i-1}}[P_i(u,M)] \mbox{ }\forall u\in N(v_i) \label{bmlp1}\\
                \sum_{u\in N(M,v_i)} x(u,M) & = &1 \mbox{ } \forall M \in S_{i-1} \label{bmlp2}\\
                x(u,M) & \geq & 0  \mbox{ } \forall M\in S_{i-1} \mbox{ }\forall u\in N(M,v_i) \label{bmlp3}
        \end{eqnarray}
        \State Set $D_i = \{ (Pr_{D_{i-1}}[M]x(u,M), M\cup \{(v_i,u)\}) : M\in S_{i-1} \} \cup \{(Pr_{D_{i-1}}[M], M) : M\in D_{i-1} \setminus S_{i-1}\}$.
        \State Remove from $D_i$ the matchings with zero probability.
        \EndFor\\
        \Return the largest matching from $D_n$.
\end{algorithmic}
\end{algorithm}

We can talk about applying Ranking given distribution $D_{i-1}$ as follows: this will yield a distribution of matchings. It is built by first taking $M \in D_{i-1}$ with the corresponding probability. Then, a random permutation consistent with $M$ is chosen, and Ranking is applied to the rest of the online vertices using this permutation. The intuition of the algorithm is that Inequality~\ref{bmlp1} ensures that, in $D_i$, the probability of using each offline vertex in the $i$th step will not exceed the probability of using it when applying Ranking to the distribution $D_{i-1}$. Equalities \ref{bmlp2} and \ref{bmlp3} ensure that we will split $Pr_{D_{i-1}}[M]$ into the $x(u,M)$'s so that $D_i$ is a probability distribution. The set of inequalities is feasible because we can set $x(u,M) = P_i(u,M)$ for all $M \in S_{i-1}$ and $u \in N(M,v_i)$. The number of inequalities in the LP, without counting inequalities~\ref{bmlp3} that state variables are non-negative, is at most $|D_{i-1}| + |N(v)|$. Thus, an extreme point solution for the $i$th LP will have at most $|D_{i-1}| + |N(v)|$ nonzero variables, so $|D_n| \leq 1 + \sum_i |N(v_i)|$: this is a $(|E|+1)$-width algorithm, where $E$ is the set of edges of the graph.

It should be noted that this algorithm is not efficient. Note that a matching defines a partial ordering on the set of offline vertices in the following way: for every matched online vertex $v$, the offline match $u$ is greater than all the other neighbours $u'$ of $v$ that were available when $u$ is matched. The number of permutations consistent with a matching is equal to the number of linear extensions of this poset: a linear ordering on the set of offline vertices can be viewed as a permutation of this set and vice-versa. Given a poset, it is easy to construct a graph and a matching that corresponds to that poset. Also, $P_i(u,M)$ is the number of linear extensions of $M \cup \{(v_i,u)\}$ divided by the number of linear extensions of $M$. So calculating $P_i(u,M')$ for every $u$ matched in $M$ ($M'$ is the subset of $M$ constructed before matching $u$) will give the number of linear extensions of the poset associated with $M$. Therefore, the $P_i(u,M)$'s calculation is $\#P$-hard (since calculating the number of linear extensions of a poset is \cite{Brightwell}).

We have not been able to prove a good guarantee for this algorithm. One reason for this is that the analysis for Ranking usually is inductive but on the offline order (according to the permutation), not the online order, so it is unclear if the inductive invariants used to analyse Ranking will remain true in this LP-based approach. In addition, $P_i(u,M)$ does not seem to be amenable to analysis. Finally, experimental results show that the expected size of matchings in $D_n$ may be worse than the expected performance of Ranking on a graph, so trying to somehow argue that the distribution obtained from the LP algorithm outperforms Ranking will not work. For instance, on the hard instance for Random (see Figure~\ref{RandHard} and its explanation), when $n=10$, the expected ratio for $D_n$ (of size of matching over maximum matching size, in this case 10) is 0.7058 while the expected ratio for Ranking on this same graph is 0.7090. On the complete upper triangular matrix (the hard instance for Ranking in \cite{Karp1990a}), when $n=6$, the expected ratio for $D_n$ is 0.6709 and the one for Ranking is 0.6761. However it should be noted that the results are fairly close and, for this family of examples, it looks like the asymptotic approximation ratio will be at least $1-1/e$.

We will first consider the bipartite matching problem in the ROM.  We will then show that the use of randomization n the algorithm or the randomization of the inout stream cannot be 
replaced by a judiciious but deterministic ordering of the online vertices nor can it be replaced by a constant (or even $O(\frac{\log n}{\log \log n})$ width of parallel online algorithms. We will then relate the width inapproximation result to previously known results for the advice and streaming models.

\subsection{Priority Inapproximation Bound}

We now turn to study unweighted bipartite matching in the priority model. 
All of our results are for adaptive priority. While there are two related models where priority bipartite matching may be studied: one-sided (where there is one data item per vertex in the online side, and the offline side is known in advance) and two-sided (where there is one data item per vertex in the graph), we shall restrict attention to the one sided model.

\subsubsection{The inapproximation for deterministic priority algorithms}

The following theorem shows that deterministic priority algorithms cannot achieve a non-trivial asymptotic approximation ratio in the one-sided model. Thus, being able to choose the order in which to process the vertices is not sufficient to overcome the lack of randomness.

\begin{theorem}\label{Prio1/2}
There does not exist a deterministic priority algorithm that achieves an asymptotic approximation ratio greater than $1/2$ for online one-sided bipartite matching, even if the algorithm knows the size of the graph.
\end{theorem}
\begin{proof}
We describe a game between an algorithm and an adversary that, for any odd integer $n \geq 3$, yields an inapproximation of $\frac{(n+1)/2}{n}$, which can be made arbitrarily close to $1/2$ by making $n$ sufficiently large. The bipartite graph will have two sides, each with $n$ vertices. Let $OFF$ be the set of offline vertices and let $ON$ be the set of online vertices. We restrict our attention to graphs in which all vertices in $ON$ have degree $(n+1)/2$. The idea is that the algorithm does not know, a priori, anything about the degrees of vertices in $OFF$, so we can adjust the neighbours of vertices in $OFF$ to ensure that the algorithm makes mistakes at each step.

The adversary will keep track of $M$, $U$, $R$, which are pairwise disjoint subsets of $OFF$ (all initially empty). $M$ will be the set of vertices in $OFF$ matched by the algorithm, $U$ will be the set of vertices that the algorithm cannot possibly match because of the matches it has already done, and $R$ will be the set of vertices that the algorithm cannot match because of the rejections it has made. Note that the algorithm won't be able to match vertices in both $U$ and $R$. The only difference between the sets is the reason for this ``unmatchability''. As the game between the algorithm and adversary progresses, the adversary announces some information about the graph, so that the set of possible instances may be further restricted.

The adversary will ensure that $|M| = |U|$ whenever the algorithm has to provide an ordering of data items. The set $P$ of potential data items is defined as the set of vertices of degree $(n+1)/2$ where the set of neighbours $N$ satisfies $M \subseteq N$ and $N \subseteq OFF \setminus (U \cup R)$. In other words, $N$ contains $M$ and is disjoint from $U$ and $R$. Initially $P$ consists of all data items of vertices of degree $(n+1)/2$ with neighbours in $OFF$, and $P$ shrinks every time $M$, $U$, and $R$ are updated. While $|U|+|R| < (n-1)/2$, the algorithm receives the data item of the vertex $v$ from $P$ that comes first in some ordering $\pi$ of data items. Note that the number of neighbours in $OFF \setminus (U \cup R \cup M)$ is at least 2 because $|M| <(n-1)/2 =(n+1)/2-1$. There are two options: the algorithm matches $v$ to a neighbour (in $OFF \setminus (U \cup R \cup M)$, because vertices in $M$ have already been matched by the algorithm), or it rejects $v$. We now show how we maintain our invariant ($|M| = |U|$) in either case.

If the algorithm matches $v$ to some vertex $m \in OFF \setminus (U \cup R \cup M)$, the adversary updates $M$ by adding $m$. This means that all the vertices in $ON$ that have not been processed yet will be neighbours of $m$. It picks another neighbour $u$ of $v$ from the set $OFF \setminus (U \cup R \cup M)$ and updates $U$ by adding $u$. This implies that none of the vertices in $ON$ that have not been processed yet will be neighbours of $u$, so it is impossible for the algorithm to match vertex $u$. If the algorithm rejects $v$, the adversary updates $R$ by adding a neighbour $r$ of $v$ from $OFF \setminus (U \cup R \cup M)$. Thus, none of the vertices in $ON$ that have not been processed yet are neighbours of $r$ and $r$ will remain unmatched. Note that, in either case, the condition $|M| = |U|$ is still maintained.

Each time a vertex in $ON$ is examined, exactly one of $|U|$ and $|R|$ is increased by 1. Thus, after $(n-1)/2$ vertices in $ON$ have been examined, $|U|+|R|=(n-1)/2$. At this point, $P$ consists of vertices whose set of neighbours is $OFF \setminus (U \cup R)$. This is necessary to guarantee that the potential data items correspond to vertices of degree  $(n+1)/2$. The adversary no longer shrinks $P$ while the remaining $(n+1)/2$ vertices are examined. The matching obtained by the algorithm is at most $n - (|U|+|R|) = (n+1)/2$ because it does not match any vertices in $U \cup R$.

However, there exists a perfect matching. We construct it by looking back at the game between the algorithm and adversary. We match the first $(n-1)/2$ vertices processed by the algorithm to vertices in $U \cup R$. For each step in which the algorithm rejected a vertex $v$, there was some vertex $r$ added to $R$: we match $v$ to $r$. For each step in which the algorithm matched a vertex $v$, it was matched to a vertex $m$ and there was a $u$ that was added to $U$. We match $v$ to $u$. Thus, after $(n-1)/2$ vertices, we have matched all of $U$ and $R$. We are left with vertices in $OFF \setminus (U \cup R)$, which can be matched in any way to the $(n+1)/2$ remaining vertices in $ON$ because that is precisely their set of neighbours.

For an example, take $n=7$. Then the degrees have to be $4$ and the adversary constructs $M$, $U$, and $R$ when the first 3 online vertices are being examined. An example where the algorithm matches the first online vertex, rejects the second, and matches the third is shown in Figure~\ref{7Bip-1}. The online vertices examined after an offline vertex is added to $M$ will be neighbours of this vertex. The online vertices examined after an offline vertex is added to $U \cup R$ will not be neighbours of this vertex. A perfect matching for this example is shown in Figure~\ref{7Bip-2}.

\begin{figure}[!ht]
  \centering
    \includegraphics[width=0.5\textwidth]{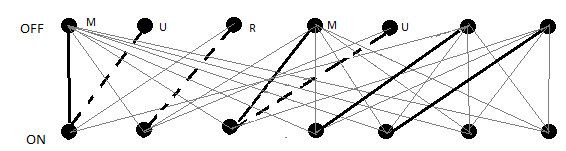}
    \caption{Offline vertices are labelled whenever they belong to $M$, $U$, or $R$. Dark lines correspond to matches made by the algorithm. Dotted lines correspond to edges such that, when the online vertex was processed, the adversary added the offline vertex to $U \cup R$. Light lines are all other edges. The algorithm obtains a matching of size 4.}
    \label{7Bip-1}
\end{figure}
\begin{figure}[!ht]
  \centering
    \includegraphics[width=0.5\textwidth]{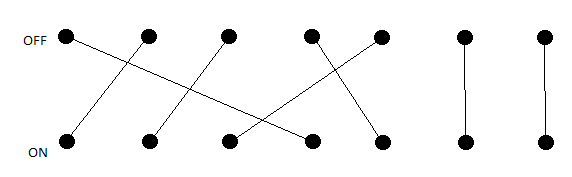}
    \caption{Same instance as before, now showing only the edges corresponding to a perfect matching obtained by the adversary.}
    \label{7Bip-2}
\end{figure}
\end{proof}


\subsubsection{Randomized priority algorithms}
\label{subsec:randomized-priority}
For randomized priority algorithms, we thus far have only  a much weaker $\frac{53}{54}$ inapproximation result in Theorem ~\ref{thm:randomized-priority-inapprox}. Algorithm Ranking can obviously be implemented in this model, and this already an approximation ratio of $1-1/e$, which is optimal for online randomized algorithms. This leaves open 
the question of whether there is a randomized priority algorithm beating the optimal randomized online
Ranking algorithm. 

The following argument emphasizes the difficulty in proving 
inapproximation results for randomized priority algorithms. Using the Yao minimax lemma, it is sufficient to define a distribution on inputs such that every 
deterministic algorithm will suffer an inapproximation (in expectation) for
that distribution.   
Consider the graph in 
Figure~\ref{PrioOneRand} and consider the following reasoning. Given that 
all online nodes have degree 2, it would 
seem that if $v$ is the first node being seen (in the random ordering of the 
online vertices), then it will fail to choose $u$ as its match with
probability $\frac{1}{2}$ and hence the probability that the algorithm makes
a mistake for the first online vertex it sees is 
$\frac{1}{6} > \frac{1}{9}$.
This would yield the weak but improved  bound of 
$\frac{17}{18}$. This reasoning is correct if the
algorithm can only make decisions for an online node based on the degree of
the node (and any previously seen information). For the first online node, 
there is no previous information, so that a  {\it degree-based} algorithm 
would indeed achieve the improved $\frac{17}{18}$ approximation. 

%

However, there is a way to order the vertices so that the probability of making a mistake for the first online vertex is less than $\frac{1}{6}$. Let a data item be described by $l_1 : (l_2, l_3)$, meaning that this corresponds to the data item of the online vertex with label $l_1$, having the offline vertices labelled with $l_2$ and $l_3$ as neighbours. Let the priority ordering given by the algorithm begin as follows: $v_1:(u_1,u_2), v_1:(u_1,u_3), v_2:(u_1,u_2), v_2:(u_1,u_3), v_3:(u_1,u_2), v_3:(u_1,u_3)$. At least one of the data items must be in the instance, regardless of how the nodes are permuted. Let us assume that upon receiving its first data item, the algorithm matches the online vertex to $u_1$.

We now analyse the probability that $v$ is matched, but not to $u$. If $v$ is not the first online vertex received, it is easy to see that the algorithm can achieve a perfect matching. When $u$ is labelled as $u_1$, then $v$ will be the first online vertex received, but the algorithm will match $v$ to $u$ and achieve a perfect match. So the only case when the algorithm can make a mistake is when $v$'s other neighbour, $u'$, is labelled as $u_1$. Since $u'$ has all online nodes as neighbours, the one labelled $v_1$ will be received first. To make a mistake, then, $u'$ has to be labelled as $u_1$ and $v$ has to be labelled as $v_1$. Thus, the probability that the algorithm makes a mistake is $\frac{1}{9}$, and the algorithm achieves an approximation ratio of $\frac{8}{9} + \frac{1}{9} \times \frac{2}{3} = \frac{26}{27}$. This algorithm shows why the above argument is incorrect: it was able to receive $v$ with probability $4/9$, and conditioned on $v$ being the first vertex, $v$ is matched to $u'$ with probability $1/4$.

\begin{figure}[!ht]
  \centering
    \includegraphics[width=0.5\textwidth]{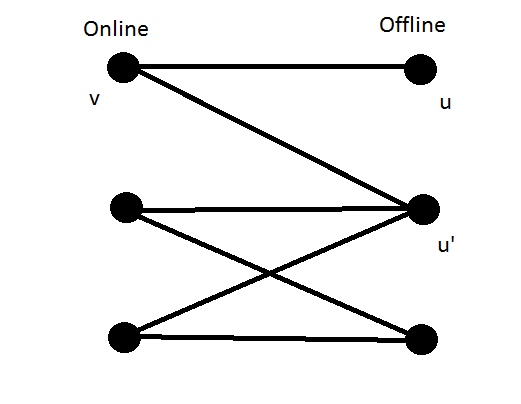}
    \caption{A graph showing the difficulty of proving priority randomized inapproximations.}
    \label{PrioOneRand}
\end{figure}   

We are able to show that randomized priority algorithms cannot achieve optimality 

\begin{theorem}
\label{thm:randomized-priority-inapprox}

No randomized priority algorithm can acheive an approximation better than 
$\frac{53}{54}$. 
\end{theorem}

\begin{proof}
Using the names as above, suppose the algorithm first considers $v_i:(v_j,v_k)$ and without loss of generality suppose the algorithm will match $v_i$ to $v_j$. Then with probability $\frac{1}{3}$, $v$ is named $v_1$, and with probability $\frac{1}{6}$, $(u',u)$ is named $(u_j,u_k)$ where $u'$ is the other neighbor of $v$.

\end{proof}

\subsection{ROM bipartite Matching}

Ranking is an online randomized algorithm,  but as observed in
\cite{Karp1990a}, it also has a well-known interpretation as a ROM algorithm. It is equivalent to the Fixed Ranking algorithm in ROM: performing online Ranking on a graph $G=(V_1,V_2,E)$ (where $V_1$ is the online side and $V_2$ the offline side) is the same as deterministically matching to the first available neighbour in the graph $G'=(V_2,V_1,E)$, in the ROM model (now $V_2$ are the online vertices, and the ordering of $V_1$ used to decide which vertex is ``first'' corresponds to the online order in $G$). This means that deterministic ROM algorithms can achieve an approximation ratio of $1-1/e$.

It is known that algorithm Random does not get an approximation ratio better than 1/2 in the online setting \cite{Karp1990a}. Figure~\ref{RandHard} shows the hard instance used for $n=6$. The online vertices are columns, and they arrive from right to left. The offline vertices are rows, and a $0$ entry means there is no edge between the corresponding online and offline vertices, while a $1$ means that there is. This is generalized in an obvious way to an instance of $2k$ vertices. The first $k$ online vertices will have degree $k+1$. They will have $k$ common neighbours (corresponding to the first $k$ rows in the matrix), plus an additional neighbour that has degree 1 (hence, the correct choice is to match the online vertex to this additional neighbour). The last $k$ online vertices will have degree 1, each being matched to one of the first $k$ rows.

\begin{figure}[!hb]
\caption{Graph that is hard for algorithm Random.}
\label{RandHard}
$$\begin{pmatrix}
1&0&0&1&1&1\\
0&1&0&1&1&1\\
0&0&1&1&1&1\\
0&0&0&1&0&0\\
0&0&0&0&1&0\\
0&0&0&0&0&1\\
\end{pmatrix}
$$\end{figure}

Random has trouble with this instance because it does not use information from previous online vertices when it has to ``guess'' the correct match of a given online vertex. In contrast, by using a random permutation of the offline vertices, Ranking tends to be biased in favor of vertices that have occurred less in the past. Suppose that several online vertices have been matched. In the permutation used by the algorithm, the remaining offline vertices from among the first $k$ must appear after all the offline vertices that have already been matched. On the other hand, we know nothing about the relative order among the first $k$ offline vertices of the offline vertex $u$ of degree 1 that is neighbour of the next online vertex. Thus, there are more permutations consistent with the choices made so far where $u$ comes first among those available, so $u$ is more likely to be matched. In contrast, in Random, at any point in time, all available vertices are equally likely to be matched.

Ranking has been well studied in ROM. It is interesting to note, however, that this is not the case for algorithm Random. Of course, it is very plausible that Random will not beat Ranking in ROM. However, it would be nice to know if it gets an approximation ratio greater than 1/2 in this model. In fact, we do not even know whether Ranking is asymptotically better than Random. It could even be the case that Random is asymptotically better than Ranking. In the analysis of the following theorem we show that the performance of Random in ROM for the instance described above is worse than that of Ranking. 

\begin{theorem}
Algorithm Random achieves an asymptotic approximation ratio of at most $3/4$ in ROM.
\end{theorem}
\begin{proof}
Consider the hard graphs of Figure~\ref{RandHard}. Fix a permutation of online vertices. Ranking will randomly choose a permutation of the offline vertices $\sigma$. Notice that Ranking and Random behave the same on online vertices of degree 1: they will match to its neighbour if it is still available. In particular, the behaviour of Ranking on degree 1 vertices reveals nothing about $\sigma$. We now see what happens to the vertices of degree $k+1$. For the first of these, both algorithms behave the same: the probability of picking any at random is equal to the probability of any being first (among the options) in the ranking.

Now, consider any online vertex $u$ of degree $k+1$ and fix the choices of matches of previous vertices. We will analyse the behaviour of Ranking and Random under the assumption that these matches were made previously by both algorithms. Let $A$ be the set of offline vertices of degree $k+1$ that were previously matched to online vertices of degree $k+1$: the mistakes from the past. Let $B$ be the set of offline vertices of degree $k+1$ that are still available currently, when choosing the match of $u$. Notice that all vertices in $B$ must have also been available before. Because Ranking chose matches with vertices in $A$, then all vertices in $A$ go before all vertices in $B$ in $\sigma$. However, consider the unique offline vertex $v$ of degree 1 that is neighbour of $u$: the correct match of $u$. This vertex could appear anywhere in $\sigma$, since it has never been considered before. This means that the probability of this vertex being first in $\sigma$ among the available vertices is greater than that of any vertex in $B$. So the probability that Ranking matches $u$ to $v$ is higher than the probability that it matches $u$ to any vertex in $B$.

Thus, for any permutation of online vertices and any fixed choices of matching for previous vertices, the probability that Ranking matches the current online vertex of degree $k+1$ to its neighbour of degree 1 (its correct match) is greater than or equal to the probability that Random does so. The size of the matching is equal to $k$ plus the number of correct matchings: any incorrect match means there will be an online degree 1 vertex that will be unmatched. It is proved in \cite{Karande2011} that the asymptotic approximation ratio of Ranking on these graphs is at most 3/4. Thus, this bound also holds for Random.
\end{proof}

We now turn our attention to deterministic ROM algorithms. The following shows that in ROM an inapproximation result using a specific small graph will not yield the same inapproximation for arbitrarily large graphs just by taking the union of disjoint copies of the small graph. Consider the bipartite graph with 2 vertices on each side, where one online vertex has degree 2 and the other has degree 1. By carefully choosing the neighbour of the degree 1 vertex, we get that no deterministic algorithm can get an approximation ratio better than 0.75. Now consider a graph with 4 vertices, with 2 components, each with the form of the graph with 2 online vertices (see Figure~\ref{DisjointCopyROM}).

\begin{figure}[!htbp]
  \centering
    \includegraphics[width=0.4\textwidth]{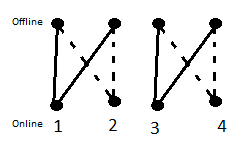}
    \caption{Vertices 2 and 4 have degree one, and the dotted lines show the possible neighbours.}
    \label{DisjointCopyROM}
\end{figure}

Consider the following algorithm: match greedily at each step (so that if vertex 2 comes before 1 and 4 before 3, a perfect matching is obtained). Now, if the first vertex received (among 1,2,3,4) is vertex 1 or 3, match it to the leftmost neighbour. Otherwise, whenever there is a choice for vertex 1 because it occurred before vertex 2, choose the rightmost neighbour (analogously for vertex 3). We show that this algorithm achieves an approximation ratio of 0.875 on these graphs. If the graph consists of vertex 2 having the leftmost offline vertex as neighbour, the algorithm will get both vertices matched whenever vertex 1 occurs first in the permutation, and this happens with probability 1/4. The probability that vertex 2 comes before 1 is 1/2. Therefore the algorithm matches $1/2*2+1/4*2+1/4*1 = 1.75$ edges in expectation. If the graph consists of vertex 2 being matched to the option on the right, the algorithm will match both when 1 does not occur first, which happens with probability 3/4. Thus in this case the algorithm matches $3/4*2+1/4*1 = 1.75$ edges in expectation. The analysis for the other component can be done in the same way. That is, the performance of the algorithm on any such graph is $2*1.75 = 3.5$. Thus, the algorithm achieves an approximation ratio of $3.5/4 = 0.875$.

To the best of our knowledge, it is not known whether there is any deterministic algorithm that achieves an asymptotic approximation ratio better than $1-1/e$. While we cannot answer this question, we do provide some experimental results suggesting that there may be some better algorithms. As in the online world, it is easy to see that we can assume without loss of generality that the algorithm is greedy: it matches whenever possible. An arbitrary algorithm $ALG$ can be simulated by a greedy algorithm $ALG'$ as follows: if $ALG$ matches online vertex $u$ to $v$ and $v$ is not yet matched in $ALG'$, do that same matching; else match $u$ to any unmatched neighbour in $ALG'$, if there is any. The set of offline vertices matched by $ALG'$ will always contain the set of offline vertices matched by $ALG$, so the approximation ratio of the former is at least as good as that of the latter.

Thus the question is, given a vertex, what offline neighbour is chosen as the match, among those available. We define the ranking of an online vertex $v$ as an ordering of offline vertices which determines which one is matched to $v$, among the ones that are available. For instance, in Fixed Ranking, the ranking used is the same for every online vertex. We consider some algorithms which differ from this greedy algorithm with a fixed ranking. In the following, we use the convention that online vertices are labelled $0, 1, \ldots, n-1$ in order of arrival, and offline vertices are somehow labelled also from $0$ to $n-1$. We consider the following algorithms:

\begin{algorithm}[!t] \caption{Cyclic Ranking}
\begin{algorithmic}[1]
	\For{$i=0$ to $n-1$}
	\State Let $v$ be the $i$th online vertex and let $d = \mbox{deg}(v)$ (counting matched offline vertices).
	\State Let $f = (i+d) \mod (n-1)$. Let $\sigma = (f, (f+1) \mod n, \ldots, (f+n-1) \mod n)$.
	\State The algorithm matches $v$ using the ranking given by $\sigma$ (first try $f$, if not try $(f+1) \mod n$, etc.).
	\EndFor
\end{algorithmic}
\end{algorithm}

\begin{algorithm}[!t] \caption{Left-right Ranking}
\begin{algorithmic}[1]
	\For{$i=0$ to $n-1$}
	\State Let $v$ be the $i$th online vertex and let $d = \mbox{deg}(v)$ (counting matched offline vertices).
	\State Let $\sigma =
	\left\{
	\begin{array}{ll}
		(0,1,\ldots,n-1)  & \mbox{if } i+d \mod 2 = 0 \\
		(n-1,n-2,\ldots,0) & \mbox{otherwise}
	\end{array}	
	\right.$.
	\State The algorithm matches $v$ using the ranking given by $\sigma$.
	\EndFor
\end{algorithmic}
\end{algorithm}

\begin{algorithm}[!t] \caption{Least-seen}
\begin{algorithmic}[1]
	\For{$i=0$ to $n-1$}
	\State Let $v$ be the $i$th online vertex.
	\State The algorithm matches $v$ to the neighbour that has occurred the least so far (breaking ties arbitrarily).
	\EndFor
\end{algorithmic}
\end{algorithm}

For every algorithm $\mathbb{A}$ and every $n$, we consider every possible bipartite graph with $n$ online vertices and $n$ offline vertices. For each graph, we consider all permutations of the online vertices. This yields a ratio $r(\mathbb{A}, n)$: the minimum over all the graphs of the average over all permutations of the performance of the algorithm on the permuted graph, divided by the optimal matching size. Recall that for $n=2$, $0.75$ is optimal for deterministic algorithms. Because the number of bipartite graphs grows exponentially and for each graph all permutations are considered, we limit ourselves to $n \leq 6$ due to the computational resources needed. Table \ref{detROM} contains the results.

\begin{table}[h!]
\centering
\caption{Experimental results for deterministic ROM algorithms}
\label{detROM}
\begin{tabular}{l|llll}
\textbf{n} & \textbf{Fixed Ranking} & \textbf{Cyclic Ranking} & \textbf{Left-right Ranking} & \textbf{Least-seen} \\ \hline
\textbf{3} & 0.7222                 & 0.7222                  & 0.7778                      & 0.7222              \\
\textbf{4} & 0.6979                 & 0.7292                  & 0.7292                      & 0.6875              \\
\textbf{5} & 0.6850                 & 0.7100                  & 0.7267                      & 0.6817              \\
\textbf{6} & 0.6762                 & 0.7023                  & 0.7069                      & 0.6722             
\end{tabular}
\end{table}

The table shows how Cyclic Ranking and Left-right Ranking seem to be beating Fixed Ranking for all $n$ (except possibly tying for $n=3$). Despite having lower performance than Fixed Ranking, we decided to include Least-seen because it seemed like a natural candidate, giving priority to offline vertices that have not been seen much, because most likely they have lower degree and thus less chance of being matched and because this somewhat corresponds to our intuition about what gave Ranking an advantage over Random in the online world. However the results seem to indicate it does not do very well against worst-case input. The most interesting question is whether there is an algorithm with asymptotic approximation ratio greater than $1-1/e$, but the fact that Fixed Ranking is not optimal for every $n$ gives hopes for a positive answer.
For unweighted matching, 
without loss of generality we can restrict attention to greedy algorithms, that is those that will always make a match whenever it encounters a data item where the vertex has available neighbours. 
  

\section{Conclusion and Open Problems}
\label{sec:conclusions}

We considered bipartite matching in the width, priority and ROM  models. By observing that non-uniform max-of-$k$ and advice results are equivalent, it follows that polynomially many online matchings are enough to achieve an approximation ratio arbitrarily close to $1-1/e$ while exponentially many online matchings are required to  achieve an approximation ratio greater than $1-1/e$. We turned to the question as to how large $k$ needs to be so that max-of-$k$ or width-$k$ algorithms can obtain an approximation ratio greater than $1/2$. We showed that if $k$ is constant then width-$k$ algorithms cannot achieve this. By analysing the proof more closely, we noticed that in fact max-of-$\frac{\log n}{\log \log n}$ is not enough. This gave new lower bounds for bipartite matching with advice. However, there is still a gap between our inapproximation bound and the non-uniform max-of-$n^{O(1)}$ algorithm that de-randomizes Ranking (see Theorem~\ref{WidthUppLow}). This algorithm is unsatisfactory because it requires exponential time pre-calculation. 
The main open problem then is whether or not a uniform and efficient polynomial width online algorithm is possible for online bipartite matching. We offered some possibilitiees in this regard. The study of width bipartite matching in the PBT model \cite{Alekhnovich2011a} is left as future work.

With regard to the ROM and priority models, we showed that deterministic priority algorithms cannot improve upon the $\frac{1}{2}$ approximation achieved by any greedy maximal matching algorithm for matching. The situation for randomized priority algorithms remains quite unclear as to whether or not the KVV randomized online Ranking algorithm $1-\frac{1}{e}$ approximation can be improved upon. We give an example showing why the analysis of such algorithms is subtle. It is also the case that biparttite online matching is not completely  understood in the ROM model. In particular, what is the precise approximation for the randomized KVV algorithm in the ROM model and is it an optimal algorithm in this regard? And what is the approximation ratio for the Random algorithm in the ROM model?

Another problem that has been studied extensively and has important theoretical and practical interest is max-sat. We considered this problem in the width models. The problem itself has multiple possible input models, and it has generalizations such as submodular max-sat. In input model 1, where the algorithm knows the weights and lengths of clauses in which the variable being processed appears (both positively and negatively), we gave an efficient linear width algorithm achieving a $3/4$ approximation ratio by using very similar techniques to those in the deterministic algorithm with approximation ratio $1/2$ for unconstrained submodular maximization \cite{Buchbinder2016}. We showed that, to achieve an approximation ratio greater than $3/4$, exponential width (in the width-cut model) is required, even with input model 2. In input model 0, where the algorithm only knows the weights of clauses in which the variable being processed appears, we showed that deterministic constant width cannot achieve an approximation ratio of $2/3$, so having either the length of clauses or randomization is essential; Johnson's deterministic algorithm and the randomized algorithm that assigns a variable with probabilities proportional to weights both achieve an approximation ratio of 2/3 \cite{Johnson} \cite{Azar2011a}. In the strongest input model, where the algorithm receives a complete description of the clauses in which the variable appears, we showed an exponential lower bound on the width-cut necessary to get an approximation ratio greater than $5/6$. This should be contrasted with the exponential lower bound on the width needed by a pBT algorithm to get an approximation ratio better than $21/22$ \cite{Alekhnovich2011a}. For input model 3, there is clearly a large gap between the online polynomial width $\frac{3}{4}$ approximation and the proven inapproximation for the pBT width model. The question of limitations of randomized priority algorithms has not been studied to the best 
of our knowledge.

Two (and multiple) pass algorithms provide another interesting extension of online (or priority) algorithms. The determinstic two-pass 
online algorithms for 
max-sat (\cite{Polo2pass}) and biparitite matching (\cite{Durr2016a}) provably improves upon the known bounds for deterministic one pass  algorithms against 
an adevrsarial input.  
 
The ultimate goal is to see to what extent ``simple combinatorial algorithms'' can achieve approximation ratios close to known hardness results and in doing so can come close to or even improve upon the best known offline algorithms.
Part of this agenda is to better understand when and how one can de-randomize 
online algorithms.

\bibliographystyle{siam}
\bibliography{thesis.bib}{}

\begin{thebibliography}{10}

\bibitem{Aggarwal}
{\sc G.~Aggarwal, G.~Goel, C.~Karande, and A.~Mehta}, {\em Online
  vertex-weighted bipartite matching and single-bid budgeted allocations},
  SODA,  (2011), pp.~1253--1264.

\bibitem{Alekhnovich2011a}
{\sc M.~Alekhnovich, A.~Borodin, J.~Buresh{-}Oppenheim, R.~Impagliazzo,
  A.~Magen, and T.~Pitassi}, {\em Toward a model for backtracking and dynamic
  programming}, Computational Complexity, 20 (2011), pp.~679--740.

\bibitem{Angelopoulos2010a}
{\sc S.~Angelopoulos and A.~Borodin}, {\em {Randomized priority algorithms}},
  Theoretical Computer Science, 411 (2010), pp.~2542--2558.

\bibitem{Aronson95}
{\sc J.~Aronson, M.~Dyer, A.~Frieze, and S.~Suen}, {\em {Randomized greedy
  matching II}}, Random Structures and Algorithms, 6 (1995), pp.~55--73.

\bibitem{Avidor2005}
{\sc A.~Avidor, I.~Berkovitch, and U.~Zwick}, {\em Improved approximation
  algorithms for max nae-sat and max sat}, WAOA,  (2005), pp.~27--40.

\bibitem{Azar2011a}
{\sc Y.~Azar, I.~Gamzu, and R.~Roth}, {\em Submodular max-sat}, ESA,  (2011),
  pp.~323--334.

\bibitem{Bahmani2010}
{\sc B.~Bahmani and M.~Kapralov}, {\em Improved bounds for stochastic online
  matching}, ESA,  (2010), pp.~170--181.

\bibitem{Besser2015}
{\sc B.~Besser and M.~Poloczek}, {\em Greedy matching : guarantees and
  limitations}, arXiv,  (2015).

\bibitem{Birnbaum2008}
{\sc B.~Birnbaum and C.~Mathieu}, {\em {On-line bipartite matching made
  simple}}, ACM SIGACT News, 39 (2008), pp.~80--87.

\bibitem{Bock11}
{\sc H.-J. B\"{o}ckenhauer, D.~Komm, R.~Kr\'alovi\v{s}, and R.~Kr\'alovi\v{s}},
  {\em On the advice complexity of the k-server problem}, ICALP,  (2011),
  pp.~207--218.

\bibitem{Borodin2010a}
{\sc A.~Borodin, J.~Boyar, K.~S. Larsen, and N.~Mirmohammadi}, {\em {Priority
  algorithms for graph optimization problems}}, Theoretical Computer Science,
  411 (2010), pp.~239--258.

\bibitem{Borodin2003}
{\sc A.~Borodin, M.~N. Nielsen, and C.~Rackoff}, {\em {(Incremental) priority
  algorithms}}, Algorithmica, 37 (2003), pp.~295--326.

\bibitem{Brightwell}
{\sc G.~Brightwell and P.~Winkler}, {\em Counting linear extensions}, Order, 8
  (1991), pp.~225--242.

\bibitem{Buchbinder2016}
{\sc N.~Buchbinder and M.~Feldman}, {\em {Deterministic algorithms for
  submodular maximization problems}}, SODA,  (2016), pp.~392--403.

\bibitem{Buchbinder2012}
{\sc N.~Buchbinder, M.~Feldman, J.~S. Naor, and R.~Schwartz}, {\em A tight
  linear time (1/2)-approximation for unconstrained submodular maximization},
  FOCS,  (2012), pp.~649--658.

\bibitem{Chen1999a}
{\sc J.~Chen, D.~K. Friesen, and H.~Zheng}, {\em {Tight bound on Johnson's
  algoritihm for maximum satisfiability}}, Journal of Computer and System
  Sciences, 58 (1999), pp.~622--640.

\bibitem{Costello2011}
{\sc K.~Costello, A.~Shapira, and P.~Tetali}, {\em {Randomized greedy: new
  variants of some classic approximation algorithms}}, SODA,  (2011),
  pp.~647--655.

\bibitem{Davis2009}
{\sc S.~Davis and R.~Impagliazzo}, {\em Models of greedy algorithms for graph
  problems}, Algorithmica, 57 (2009), pp.~269--317.

\bibitem{Devanur2013a}
{\sc N.~R. Devanur, K.~Jain, and R.~D. Kleinberg}, {\em Randomized primal-dual
  analysis of ranking for online bipartite matching}, SODA,  (2013),
  pp.~101--107.

\bibitem{Durr2016a}
{\sc C.~D{\"{u}}rr, C.~Konrad, and M.~Renault}, {\em {On the power of advice
  and randomization for online bipartite matching}}, to appear in ESA 2016,
  (2016).

\bibitem{EmekFKR11}
{\sc Y.~Emek, P.~Fraigniaud, A.~Korman, and A.~Ros{\'{e}}n}, {\em Online
  computation with advice}, Theor. Comput. Sci., 412 (2011), pp.~2642--2656.

\bibitem{Engebretsen2004a}
{\sc L.~Engebretsen}, {\em {Simplified tight analysis of Johnson's algorithm}},
  Information Processing Letters, 92 (2004), pp.~207--210.

\bibitem{FeigenbaumKMSZ}
{\sc J.~Feigenbaum, S.~Kannan, A.~McGregor, S.~Suri, and J.~Zhang}, {\em On
  graph problems in a semi-streaming model}, Theor. Comput. Sci., 348 (2005),
  pp.~207--216.

\bibitem{Feldman2009a}
{\sc J.~Feldman, A.~Mehta, V.~Mirrokni, and S.~Muthukrishnan}, {\em {Online
  stochastic matching: Beating 1-1/e}}, FOCS,  (2009), pp.~117--126.

\bibitem{GoelKK12}
{\sc A.~Goel, M.~Kapralov, and S.~Khanna}, {\em On the communication and
  streaming complexity of maximum bipartite matching}, in Proceedings of the
  Twenty-Third Annual {ACM-SIAM} Symposium on Discrete Algorithms, {SODA} 2012,
  Kyoto, Japan, January 17-19, 2012, 2012, pp.~468--485.

\bibitem{GoelM08}
{\sc G.~Goel and A.~Mehta}, {\em Online budgeted matching in random input
  models with applications to adwords}, in Proceedings of the Nineteenth Annual
  {ACM-SIAM} Symposium on Discrete Algorithms, {SODA} 2008, San Francisco,
  California, USA, January 20-22, 2008, 2008, pp.~982--991.

\bibitem{Iwama2000}
{\sc M.~M. Halld\'orsson, K.~Iwama, S.~Miyazaki, and S.~Taketomi}, {\em Online
  independent sets}, COCOON,  (2000), pp.~202--209.

\bibitem{Hastad01}
{\sc J.~H\'astad}, {\em Some optimal inapproximability results}, J. {ACM}, 48
  (2001), pp.~798--859.

\bibitem{HopcroftK73}
{\sc J.~E. Hopcroft and R.~M. Karp}, {\em An n\({}^{\mbox{5/2}}\) algorithm for
  maximum matchings in bipartite graphs}, {SIAM} J. Comput., 2 (1973),
  pp.~225--231.

\bibitem{Huang2014}
{\sc N.~Huang and A.~Borodin}, {\em {Bounds on double-sided myopic algorithms
  for unconstrained non-monotone submodular maximization}}, ISAAC,  (2014),
  pp.~528--539.

\bibitem{Iwama2002}
{\sc K.~Iwama and S.~Taketomi}, {\em Removable online knapsack problems},
  ICALP,  (2002), pp.~293--305.

\bibitem{Jaillet}
{\sc P.~Jaillet and X.~Lu}, {\em Online stochastic matching: new algorithms
  with better bounds}, Math Operations Research,  (2014), pp.~624--664.

\bibitem{Johnson}
{\sc D.~Johnson}, {\em Approximation algorithms for combinatorial problems},
  STOC,  (1973), pp.~38--49.

\bibitem{Karande2011}
{\sc C.~Karande, A.~Mehta, and P.~Tripathi}, {\em Online bipartite matching
  with unknown distributions}, STOC,  (2011), pp.~587--596.

\bibitem{Karp1990a}
{\sc R.~M. Karp, U.~V. Vazirani, and V.~V. Vazirani}, {\em An optimal algorithm
  for on-line bipartite matching}, STOC,  (1990), pp.~352--358.

\bibitem{Kesselheim2013a}
{\sc T.~Kesselheim, K.~Radke, A.~T{\"{o}}nnis, and B.~V{\"{o}}cking}, {\em {An
  optimal online algorithm for weighted bipartite matching and extensions to
  combinatorial auctions}}, ESA,  (2013), pp.~589--600.

\bibitem{Madry13}
{\sc A.~Madry}, {\em Navigating central path with electrical flows: From flows
  to matchings, and back}, in 54th Annual {IEEE} Symposium on Foundations of
  Computer Science, {FOCS} 2013, 26-29 October, 2013, Berkeley, CA, {USA},
  2013, pp.~253--262.

\bibitem{Mahdian2011a}
{\sc M.~Mahdian and Q.~Yan}, {\em {Online bipartite matching with random
  arrivals: an approach based on strongly factor-revealing LPs}}, STOC,
  (2011), pp.~597--606.

\bibitem{Manshadi2011}
{\sc V.~H. Manshadi, S.~O. Gharan, and A.~Saberi}, {\em {Online stochastic
  matching: online actions based on offline statistics}}, SODA,  (2011),
  pp.~1285--1294.

\bibitem{Mehta2005}
{\sc A.~Mehta, A.~Saberi, U.~Vazirani, and V.~Vazirani}, {\em Adwords and
  generalized on-line matching}, FOCS,  (2005), pp.~264--273.

\bibitem{Mikkelsen15}
{\sc J.~W. Mikkelsen}, {\em Randomization can be as helpful as a glimpse of the
  future in online computation}, CoRR, abs/1511.05886 (2015).

\bibitem{Muthukrishnan05}
{\sc S.~Muthukrishnan}, {\em Data streams: Algorithms and applications},
  Foundations and Trends in Theoretical Computer Science, 1 (2005).

\bibitem{Poloczek2011b}
{\sc M.~Poloczek}, {\em {Bounds on greedy algorithms for MAX SAT}}, ESA,
  (2011), pp.~37--48.

\bibitem{Poloczek2011c}
{\sc M.~Poloczek and G.~Schnitger}, {\em {Randomized variants of Johnson's
  algorithm for MAX SAT}}, SODA,  (2011), pp.~656--663.

\bibitem{Polo2pass}
{\sc M.~Poloczek, G.~Schnitger, D.~Williamson, and A.~van Zuylen}, {\em Greedy
  algorithms for the maximum satisfiability problem: Simple algorithms and
  inapproximability bounds}.
\newblock \textit{In Submission.}

\bibitem{SleatorT85}
{\sc D.~D. Sleator and R.~E. Tarjan}, {\em Amortized efficiency of list update
  and paging rules}, Commun. {ACM}, 28 (1985), pp.~202--208.

\bibitem{VanZuylen}
{\sc A.~van Zuylen}, {\em Simpler 3/4-approximation algorithms for max sat},
  WAOA,  (2011), pp.~188--197.

\bibitem{Yann}
{\sc M.~Yannakakis}, {\em On the approximation of maximum satisfiability},
  SODA,  (1992), pp.~1--9.

\bibitem{Yao80}
{\sc A.~C. Yao}, {\em New algorithms for bin packing}, J. {ACM}, 27 (1980),
  pp.~207--227.

\bibitem{YungPMS}
{\sc C.~Yung}, {\em Inapproximation result of exact max-2-sat}.
\newblock \textit{Unpublished}.

\end{thebibliography}

\end{document}